\documentclass[a4paper]{IEEEtran}
\usepackage{multirow}
\usepackage{cite}
\pdfoutput=1
\usepackage{hyperref}
\hypersetup{
    colorlinks=true,
    linkcolor=black,
    citecolor=black,
    filecolor=black,
    urlcolor=black,
}
   \usepackage[pdftex]{graphicx}
   \usepackage{epstopdf}
   \usepackage{epsfig}
   \graphicspath{{eps/}{png/}}
   \DeclareGraphicsExtensions{.png,.eps}
   
   \ifCLASSOPTIONcompsoc
\usepackage[caption=false,font=normalsize,labelfont=sf,textfont=sf]{subfig}
\else
\usepackage[caption=false,font=footnotesize]{subfig}
\fi
   
\usepackage[cmex10]{amsmath}
\usepackage{array}
\usepackage{mdwmath}
\usepackage{mdwtab}
\usepackage{eqparbox}
\usepackage{stfloats}
\usepackage{url}
\newtheorem{theorem}{Theorem}
\newtheorem{lemma}{Lemma}
\newtheorem{corollary}{Corollary}

\usepackage{algorithm}
\usepackage{color}

\def\B{{\mathcal{B}}}
\def\U{{\mathcal{U}}}

\begin{document}
\title{Utility Fair Optimisation of Antenna Tilt Angles in LTE Networks}
\author{Bahar~Partov$^1$,
        Douglas J.~Leith$^1$,
        and~Rouzbeh~Razavi$^2$\\
        $^1$Hamilton Institute, NUI Maynooth, $^2$ Bell Laboratories, Alcatel-Lucent, Dublin\thanks{This material is based upon works supported by the Science Foundation Ireland under Grant No. 11/PI/1177 and by Bell Labs Ireland.}}

\maketitle

\begin{abstract}
We formulate adaptation of antenna tilt angle as a utility fair optimisation task.   This optimisation problem is non-convex, but in this paper we show that under reasonable conditions it can be reformulated as a convex optimisation.   Using this insight, we develop a lightweight method for finding the optimal antenna tilt angles, making use of measurements which are already available at base stations, and suited to distributed implementation.
\end{abstract}

\begin{IEEEkeywords}
Antenna tilt angle, LTE, Proportional fairness, Maximising capacity, Optimisation
\end{IEEEkeywords}


 \section{Introduction}

The antenna tilt angle of wireless base-stations is known to be a key factor in determining cell coverage and to play a significant role in interference management  \cite{forkel2002effect}, \cite{athley2010impact}.   While traditionally adjustment of tilt angle has largely been carried out manually, modern base stations increasingly allow automated adjustment.   This creates the potential for more dynamic adaptation of tilt angle, for example to better match cell coverage to the distribution of user equipments and traffic, to reduce coverage holes created by failures in neighbouring stations, to better manage interference from the user deployment of femtocells, \emph{etc}.   The benefits of self configuration and self optimisation are already recognised in LTE release 9  \cite{3GPPTR136902}, and automated adaptation of tilt angle in particular has been the subject of recent interest.    

In this paper we formulate adaptation of antenna tilt angle as a utility fair optimisation task.   Namely, the objective is to jointly adjust antenna tilt angles within the cellular network so as to maximise user utility, subject to network constraints.   Adjustments at base stations must be carried out jointly in a co-ordinated manner in order to manage interference.    This optimisation problem is non-convex, but in this paper we show that under certain conditions it can be reformulated as a convex optimisation.     Specifically, we show that (i) in the high \textcolor{black}{signal to interference ratio} (SINR) operating regime and with an appropriate choice of decision variables, the optimisation is convex for any concave utility function, and (ii) in any SINR regime the optimisation can be formulated in a convex manner when the objective is a proportional fair rate allocation.    Since the optimisation is not well-suited to solution using standard dual methods, we develop a primal-dual method for finding the optimal antenna tilt angles.  This approach is lightweight, making use of measurements which are already available at base stations, and suited to distributed implementation.

The rest of the paper is organized as follows. In Section \ref{sec:relatedwork} we summarize the existing work in the area.  In Section \ref{sec:networkmodel} we introduce our network model, which is 
based on 3GPP standard
, and in Section \ref{sec:convexity-high-SINR} we analyse its convexity properties in the high SINR regime.  In Section \ref{sec:proportionalfair} we extend the analysis to general SINR regimes. In Section \ref{sec:perf} we carry out a  performance evaluation of a realistic setup and finally, in Section \ref{sec:conclusion}, we summarise our conclusions.  

 
 \section{Related Work}
\label{sec:relatedwork}
The analysis and modelling of the impact of the antenna tilt angle on cell performance has been well studied, see for example \cite{benner1996effects,yilmaz2009analysis} and references therein.    Recently, self-optimisation of tilt angle has started to attract attention,  but most of this work makes use of heuristic approaches.  In \cite{conf/vtc/EckhardtKG11} a heuristic method is proposed for adjusting tilt to maximise average spectral efficiency within the network, while \cite{Razavi} proposes a combination of fuzzy and reinforcement learning.  In \cite{temesvary2009self}  simulated annealing is considered for joint self-configuration of antenna tilt angle and power and in \cite{calcev2006antenna} a non-cooperative game approach between neighbouring base stations is studied. Offline planning of tilt angle is considered, for example, in \cite{eisenblatter2008capacity}, using a heuristic search method combined with a mixed integer local search.  \textcolor{black}{In the present paper, we take a more formal, rigorous approach and show that tilt angle optimisation can, in fact, be formulated as a convex problem.  Building on this result, we then introduce a lightweight distributed algorithm based on primal-dual subgradient updates and show that this algorithm is guaranteed to converge arbitrarily closely to the network optimum.}
 
\section{Network Model}
\label{sec:networkmodel}
\subsection{Network Architecture}
\label{subsec:topology} 
The network consists of a set $\B$ of base stations and a set $\U$ of \textcolor{black}{User Equipment} (UE), with UE $u\in \U$ receiving downlink traffic transmitted from base station $b(u) \in \B$.   For base stations with sectoral antennas, we define a separate element in $\B$ for each antenna.   We denote the $(x,y)$ geographical co-ordinates of base station $b$ by $(x_b,y_b)$ and of user equipment $u$ by $(x_u,y_u)$.  The distance between user $u$ and base station $b$ is therefore given by 

\begin{IEEEeqnarray}{c}
d_{b,u} = \sqrt{(x_u-x_b)^2+(y_u-y_b)^2}
\end{IEEEeqnarray}

 \subsection{Antenna Gain and Path Loss}
\label{tilt-sinr}
The received power on sub-carrier $n$ from base station $b\in \B$ at user $u\in\U$ is given by $\tilde{G}_{b,u}(\theta_{b})\rho_{b,u}P_{b,n}$, where $\tilde{G}_{b,u}(\theta_{b})$ is the base station antenna gain, $\rho_{b,u}$ the path loss between $b$ and $u$, $P_{b,n}$ is the base station transmit power 
for sub-carrier $n$
. For simplicity, shadowing and fast fading are not considered in the equations.  We model path loss, as recommended in \cite{3GPPTR36.814V9}, by
\begin{IEEEeqnarray}{c}
\rho_{b,u} = \rho_0 d_{b,u}^{- \beta}
\end{IEEEeqnarray}
with fixed path loss factor $\rho_0$, path loss exponent $\beta$ and distance $d$ in kilometres.  For a given antenna type, the antenna gain $~\tilde{G}_{b,u}(\theta_{b})$  can be determined given the relative positions of $b$ and $u$, the antenna tilt angle $\theta_b$ and the azimuth angle $\phi_b$.  With regard to the latter, changing the tilt and/or azimuth angles  changes the direction of the antenna's main lobe, see Fig \ref{fig:angle}.   We will assume that the azimuth angle is held fixed but allow the antenna tilt angle to be adjusted within the interval $[\underline{\theta},\bar{\theta}]$.   Following~\cite{3GPPTR36.814V9}, the antenna gain can then be modelled by:

\begin{IEEEeqnarray}{c}
\label{total_Gain}
\tilde{G}_{b,u}(\theta_{b}) = \tilde{G}_0 \tilde{G}_{v}(\theta_b, d_{b,u})
\end{IEEEeqnarray} 
where $\tilde{G}_0$ is the maximum gain of the antenna, 
\begin{align}
\label{eq:antenna-pattern-no-min}
\tilde{G}_{v}(\theta_{b}, d_{b,u})
&=10^{-1.2\left(\frac{\theta_{b,u} - \theta_{b}}{\theta_{3dB}}\right)^{2}}
\end{align}
is the antenna vertical attenuation, $\theta_{b,u}=\tan^{-1}\left(h/d_{b,u}\right)$, $h$ is the height difference between the base station and UE (which, for simplicity, we assume is the same for all base stations and users) and $\theta_{3dB}$ the vertical half power beam width of the antenna.   Figure~\ref{fig:vertical gain} illustrates the ability of (\ref{eq:antenna-pattern-no-min}) to accurately model the main lobe of an antenna which is popular in cellular networks.

\begin{figure}
\centering
\includegraphics[width=0.5\columnwidth]{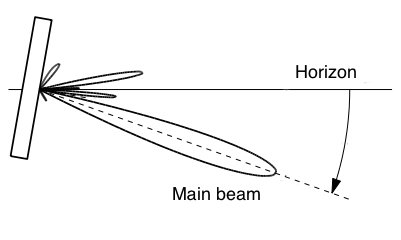}
\caption{ Schematic illustrating relationship between antenna main lobe and tilt angle.   } 
\label{fig:angle}
\end{figure}

\begin{figure}
\centering
\includegraphics[width=0.8\columnwidth]{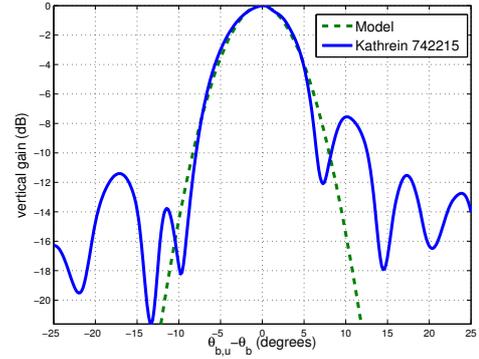}
\caption{ Comparison of antenna main lobe vertical gain model (\ref{eq:antenna-pattern-no-min}) (dashed line) and  measured antenna gain (solid line) for a Kathrein 742215 antenna, $\theta_{3dB}=9^\circ$.   } 
\label{fig:vertical gain}
\end{figure}

It will prove useful to use the quantity $G_{v}(\theta_{b}, d_{b,u}) := \log{\tilde{G}_{v}(\theta_{b}, d_{b,u})}$.   It will also prove useful to consider the following linear approximation $\hat{G}_{v}(\theta_{b}, d_{b,u})$ to antenna gain exponent $G_{v}(\theta_{b}, d_{b,u})$ about tilt angle $\theta_{0}$,
\begin{align}
\label{eq:G_V_appx}
\hat{G}_{v}(\theta_{b}, d_{b,u})= \frac{-1.2 \log10}{ \theta_{3dB}^{2}}\Big((\theta_{b,u}-\theta_{0})^{2}  + 2(\theta_{b,u}-\theta_0)\theta_{b}\Big)
\end{align}
This linear approximation is illustrated by the solid line in Figure~\ref{fig:G_v_appx}.   It is reasonably accurate provided adaptation of the antenna angle $\theta_b$ does not cause $\theta_{b,u}-\theta_b$ to change sign (in which case the side of the main antenna lobe facing the user changes and so the slope of the linear approximation changes sign).    This is assumed to be the case for the antennas of base stations other than that to which the UE is associated, which is only a mild assumption since otherwise interference from these base stations can be expected to be excessive.

\begin{figure}
\centering
\includegraphics[width=0.8\columnwidth]{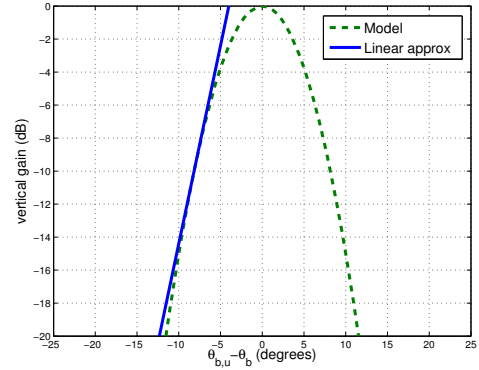}
\caption{ Illustrating linear approximation $\hat{G}_{v}(\theta_{b}, d_{b,u})$ to $G_{v}(\theta_{b}, d_{b,u})$, $\theta_{3dB}=9^\circ$.}
\label{fig:G_v_appx}
\end{figure}

 \subsection{User Throughput}
The downlink throughput of the user equipment $u \in \U$ associated with base station $b(u)$ is given by:
\begin{IEEEeqnarray}{c}
\label{eq:capacity-truncated}
R_{u}(\Theta)=\min\{\bar{r}, r_{u}(\Theta)\}, \quad u \in \U
\end{IEEEeqnarray}
where $\Theta$ is the vector $[\theta_b]$, $b\in\B$ of tilt angles, $\bar{r}$ is the maximum achievable throughput (limited by the available modulation and coding schemes), and
\begin{IEEEeqnarray}{c}
\label{eq:user-capacity-shannon}
r_{u}(\Theta)=\frac{w}{N_{sc}} \sum_{n=1}^{N_{sc}} \log(1+\kappa\gamma_{u,n}(\Theta))
\end{IEEEeqnarray}
Here $N_{sc}$ is the number of sub-carriers, $w$ the channel bandwidth, $\kappa$ a loss factor capturing non-ideal coding \emph{etc}, and  $\gamma_{u,n}(\Theta)$ SINR on sub carrier $n$ for user $u$,
\begin{IEEEeqnarray}{c}
\label{eq:dl-sinr}
\gamma_{u,n}(\Theta)=\frac{H_u(\theta_{b(u)})}{\sum_{c\in \B\setminus\{b(u)
\}}\hat{H}_u(c) +\eta_{u,n}}
\end{IEEEeqnarray}
where $H_{u}(\theta_{b}):=e^{G_{b,u}(\theta_{b})}\rho_{b,u}P_{b,u}$ is the received power from base station $b(u)$ by user $u$,  $\hat{H}_{u}(\theta_{c}):=e^{\hat{G}_{c,u}(\theta_{c})}\rho_{c,u}P_{c,u}$ is the received power from base station $c \ne b(u)$ by user $u$ and $\eta_{u,n}$ the channel noise for user $u$ on sub carrier $n$.  Observe that in $\hat{H}_{u}(\theta_{c})$ we make use of linear approximation $\hat{G}_{v}$.

 \section{ High SINR Regime}
\label{sec:convexity-high-SINR}

In the high SINR regime, the downlink throughput (\ref{eq:capacity-truncated}) can be accurately approximated by
\begin{IEEEeqnarray}{c}
\hat{R}_{u}(\Theta)=\min\{\bar{r}, \hat{r}_{u}(\Theta)\}, \quad u \in \U.
\end{IEEEeqnarray}
where
\begin{IEEEeqnarray}{c}
\label{eq:user-capcity-hat}
\hat{r}_{u}(\Theta)=\frac{w}{N_{sc}} \sum_{n=1}^{N_{sc}}\log( \kappa\gamma_{u,n}(\Theta))
\end{IEEEeqnarray} 
%

\subsection{Utility Fair Optimisation of Tilt Angle}
Under the assumption of high SINR operation, we can formulate the selection of antenna tilt angles as the following optimisation problem $(P1)$:
\begin{IEEEeqnarray}{rCl}
\max_{\Theta}  &\sum_{u\in \U} U( \hat{R}_{u}(\Theta))  \label{eq:main-objective} \\
s.t.\quad & \underline{\theta} \le \theta_b \le \bar{\theta}, \quad b\in\B \label{eq:constraint1}\\
& \underline{r} \leq \hat{R}_{u}(\Theta), \quad u \in \U \label{eq:constraint2}
\end{IEEEeqnarray}
where $U(\cdot)$ is a concave increasing utility function and $\underline{r}\ge 0$.  Constraint  (\ref{eq:constraint1}) captures restrictions on the range of feasible antenna tilt angles, while (\ref{eq:constraint2}) ensures that each user receives a specified minimum throughput (which is expected to mainly be important for users at the edge of a cell who might otherwise be assigned too low a throughput).

\subsection{Convexity Properties}

\begin{lemma}\label{th:one}
$\log H_u(\theta_{b})$ is strictly concave in $\theta_{b}$.
\end{lemma}
\begin{proof}
We have 
\begin{align*}
\log H_u(\theta_{b}) &= G_{b,u}(\theta_{b}) + \log \rho_{b,u}+\log P_{b,u}\\
& = -1.2\frac{\log10}{\theta_{3dB}^2}\left(\theta_{b,u}-\theta_b\right)^2 + \log \rho_{b,u}+\log P_{b,u}
\end{align*}
Now $\log \rho_{b,u}$, $\log P_{b,u}$, $\theta_{3dB}$ and $\theta_{b,u}$ are constants, so we only need to consider concavity with respect to $\theta_b$.  It can be verified that $\frac{d^2\log H_u(\theta_{b})}{d\theta_{b}^2}=-2.4\frac{\log10}{\theta_{3dB}^2}<0$.
\end{proof}

\begin{lemma}\label{th:two}
$\log\left(\sum_{c\in\B\setminus\{b(u)\}} \hat{H}_u(\theta_c)+\eta_{u,n}\right)$ is convex.
\end{lemma}
\begin{proof}
Rewrite as\\ $\log\left(\sum_{c\in\B\setminus\{b(u)\}} e^{\hat{G}_{c,u}(\theta_{c})+\log\rho_{c,u}P_{c,u}}+\eta_{u,n}\right)$.   This can be expressed as  $\log\left(\sum_{c\in\B\setminus\{b(u)\}} e^{h_{c,u}(\theta_c)}+\eta_{u,n}\right)$ with $h_{c,u}(\theta_c):=\hat{G}_{c,u}(\theta_{c})+\log\rho_{c,u}P_{c,u}$. By\textcolor{black}{~\cite[p74]{boyd2004convex}}, the log of a sum of exponentials is convex.  The additive term $\eta_{u,n}$ acts as a translation and by \textcolor{black}{~\cite[p79]{boyd2004convex}}  convexity is preserved under translation.  By approximation (\ref{eq:G_V_appx}), the function $h_{c,u}(\theta_c)$ is affine in $\theta_c$, and by \textcolor{black}{\cite[p79]{boyd2004convex}} when composed with the log of a sum of exponentials the resulting function remains convex.  
\end{proof}
It follows from Lemmas \ref{th:one} and \ref{th:two} that
\begin{theorem}\label{th:cor1}
$\hat{R}_u(\Theta)$ is concave in $\Theta$.      
\end{theorem}
\begin{proof}
Recalling from (\ref{eq:dl-sinr})  and (\ref{eq:user-capcity-hat}):
\begin{IEEEeqnarray}{rCl}
\hat{r}_u(\Theta) =\frac{w}{N_{sc}}\sum_{n=1}^{N_{sc}}&\bigg(&\log\kappa + \log H_u(\theta_{b(u)})\nonumber\\
&& \: - \log\big(\sum_{c\in\B\setminus\{b(u)\}} \hat{H}_u(c) + \eta_{u,n}\big)\bigg)\IEEEeqnarraynumspace
\end{IEEEeqnarray}
By Lemmas  \ref{th:one} and  \ref{th:two}, $\hat{r}_u(\Theta)$ is concave in $\Theta$.  Since the $\min$ function is concave and nondecreasing, it follows that $\hat{R}_u(\Theta):=\min\{\hat{r}_u(\Theta),\bar{r}\}$ is also concave in $\Theta$.  
\end{proof}
Note that $\hat{R}_u(\Theta)$  is not strictly concave in $\Theta$ since $\log H_u(\theta_{b(u)})$ is only strictly concave in $\theta_b(u)$ but not in the other elements of $\Theta$.  Nevertheless, under mild conditions  $\sum_{u\in\U} \hat{R}_u(\Theta)$ is strictly concave in $\Theta$:
\begin{theorem}\label{th:three}
Suppose $\B=\cup_{u\in\U} b(u)$, \emph{i.e.} every base station $b\in\B$ has at least one associated UE $u\in\U$.  Then $\sum_{u\in\U}\hat{r}_u(\Theta)$ is strictly concave in $\Theta$ (and so the solution to \textcolor{black}{problem} $(P1)$ is unique).   
\end{theorem}
\begin{proof}
We have
\begin{IEEEeqnarray*}{rCl}
\sum_{u\in\U}\hat{r}_u(\Theta) 
 = \frac{w}{N_{sc}}\sum_{n=1}^{N_{sc}}&\bigg(&\sum_{u\in\U}\log\kappa +\sum_{u\in\U}\log H_u(\theta_{b(u)}) \\
 && \: -  \sum_{u\in\U}\log(\sum_{c\in\B\setminus\{b(u)\}} \hat{H}_u(c) + \eta_{u,n})\bigg)
\end{IEEEeqnarray*}
Recall $H_u(\theta_{b(u)})$ is strictly concave in $\theta_{b(u)}$ (by Lemma \ref{th:one}).   The sum $\sum_{u\in\U}\log H_u(\theta_{b(u)})$ is therefore strictly concave in every $\theta_b$, $b\in\cup_{u\in\U} b(u)=\B$.  It is therefore strictly concave in $\Theta$ (in more detail, for any $\Theta^1$, $\Theta^2$ and $\alpha\in[0,1]$ we have $\sum_{u\in\U}\log H_u(\alpha\theta^1_{b(u)}+(1-\alpha)\theta^2_{b(u)})>\sum_{u\in\U}\left(\alpha\log H_u(\theta^1_{b(u)})+(1-\alpha)\log H_u(\theta^2_{b(u)})\right)=\alpha\sum_{u\in\U}\log H_u(\theta^1_{b(u)})+(1-\alpha)\sum_{u\in\U}\log H_u(\theta^2_{b(u)})$). The result then follows from the fact that the sum of a strictly concave function and a concave function is strictly concave.   
\end{proof}
And we have the following corollary,
\begin{corollary}\label{th:cor}
When each base station has at least one user with throughput less that $\bar{r}$, then $\sum_{u\in\U}\hat{R}_u(\Theta)$ is strictly concave in $\Theta$.    
\end{corollary}
\begin{proof}
Let $\bar{\U}\subset \U$ denote the set of users with throughput less than $\bar{r}$.   When each base station has at least one user with throughput less that $\bar{r}$ then $\B=\cup_{u\in\bar{\U}} b(u)$.   Now $\sum_{u\in\bar{\U}}\hat{R}_u(\Theta)=\sum_{u\in\bar{\U}}\hat{r}_u(\Theta)$ is strictly concave in $\Theta$ by Theorem \ref{th:three}.  It then follows immediately that $\sum_{u\in\U}\hat{R}_u(\Theta)=\sum_{u\in\bar{\U}}\hat{R}_u(\Theta) + \sum_{u\in U\setminus\bar{\U}}\bar{r}$ is strictly concave in $\Theta$.
\end{proof}

\subsection{Convex Optimisation}
\label{subsec:kkt conditions_main}

The objective function in optimisation problem (P1) is concave in $\Theta$ (since $U(\cdot)$ is concave increasing and $\hat{R}_{u}(\Theta)$ is concave by Theorem \ref{th:cor1}, then $U(\hat{R}_{u}(\Theta))$ is concave) and constraints (\ref{eq:constraint1})-(\ref{eq:constraint2}) are linear (and so convex).  Hence, optimisation problem (P1) is convex.    It follows immediately that a solution exists.    The Slater condition is satisfied and so strong duality holds.  

\subsection{Difficulty of Using Conventional Dual Algorithms}

The Lagrangian is
\begin{IEEEeqnarray}{rCl}\label{eq:lagrangian}
L(\Theta,\Lambda) &= -\sum_{u\in\U} U(\hat{R}_u(\Theta))+\sum_{u\in\U}\lambda_u^1(\underline{r}-\hat{R}_u(\Theta))  \notag \\
&+ \sum_{b\in\B}\lambda^2_b (\underline{\theta}-\theta_b) +  \sum_{b\in\B}\lambda^3_b (\theta_b-\bar{\theta})
\end{IEEEeqnarray}
where $\Lambda$ denotes the set of multiplers $\lambda_u^1$, $\lambda_b^2$, $\lambda_b^3$, $u\in\U$, $b\in\B$.  The dual function is $g(\Lambda):= L(\Theta^*(\Lambda),\Lambda)$, where $\Theta^*(\Lambda)=\arg \max_{\Theta} L(\Theta,\Lambda)$.   The main KKT conditions are $dL(\Theta,\Lambda)/d\theta_b=0$, $b \in \B$.  That is,
\begin{IEEEeqnarray}{c}\label{eq:kkt2}
\sum_{u\in\U}\left(1 +\lambda_u^1\right)\partial_{\theta_b} U(\hat{R}_u(\Theta))= \lambda_b^3- \lambda_b^2, \quad b \in \B
\end{IEEEeqnarray}
Given $\Lambda$, we can use (\ref{eq:kkt2}) to find  $\Theta^*(\Lambda)$.    The optimal vector of multipliers is $\Lambda^* = \arg \max_{\Lambda \ge 0} g(\Lambda)$.  Since $g(\Lambda)$ is concave, a standard dual function approach is to find $\Lambda^*$ using subgradient ascent techniques, and then find the optimal tilt angle $\Theta^*(\Lambda^*)$.    However, solving (\ref{eq:kkt2}) to obtain the primal variables is tricky in general since it imposes complex, implicit dual constraints for a solution to exist.   Consequently, the dual subgradient approach is unattractive for solving problem (P1).

 \subsection{Distributed Algorithm for Finding Optimal Solution}\label{sec:algo}
We consider the following primal-dual algorithm:
\begin{algorithm}[h!]
Initialise: $t=0$, $\Theta(0)$, $\Lambda(0)$, step size $\alpha>0$\\
\textbf{do}\\
\begin{minipage}[c]{\columnwidth}
\begin{align}
\theta_b(t+1) &= \theta_b(t) - \alpha \partial_{\theta_b}L(\Theta(t),\Lambda(t)),\quad \quad \  \theta_b\in\B \label{eq:dyn01}\\
\lambda_u^1(t+1) &= \left[\lambda_u^1(t) + \alpha \partial_{\lambda_u^1}L(\Theta(t),\Lambda(t))\right]^+,\  u\in\U \label{eq:dyn01a}\\
\lambda_b^i(t+1) &= \left[\lambda_b^i(t) + \alpha \partial_{\lambda_b^i}L(\Theta(t),\Lambda(t))\right]^+,\ b\in\B, i=2,3\IEEEeqnarraynumspace \label{eq:dyn02}\\
t\leftarrow t+1 \nonumber
\end{align}
\end{minipage}
\textbf{loop}
\caption{High SINR }\label{algo1}
\end{algorithm}

\noindent where in Algorithm \ref{algo1} projection $[z]^+$ equals $z$ when $z\ge0$ and $0$ otherwise, 
\begin{align}
\partial_{\theta_b}L(\Theta,\Lambda) &= -\sum_{u\in\U}\left(1 +\lambda_u^1\right)\partial_{\theta_b}  U(\hat{R}_u(\Theta)) - \lambda^2_b +  \lambda^3_b \label{eq:subgrad01}\\
\partial_{\lambda_u^1}L(\Theta,\Lambda)&= \underline{r}-\hat{R}_u(\Theta) \label{eq:subgrad01a}\\
\partial_{\lambda_b^2}L(\Theta,\Lambda)&= \underline{\theta} - \theta_b\\
\partial_{\lambda_b^3}L(\Theta,\Lambda)&=\theta_b - \bar{\theta}
\label{eq:subgrad02}
\end{align}
and $\partial_{\theta_b}  U(\hat{R}_u(\Theta))$ denotes any subgradient of $U(\hat{R}_u(\Theta))$ with respect to $\theta_b$.  

Observe that each iteration (\ref{eq:dyn01})-(\ref{eq:dyn02}) of Algorithm \ref{algo1} simultaneously updates both the primal variable $\Theta$ and the multipliers $\lambda_u^1$, $\lambda_b^2$, $\lambda_b^3$.   It possesses the following convergence property:
\begin{lemma}\label{th:converge}
For Algorithm 1 suppose  $(\Theta(t),\Lambda(t))$ is bounded for all $t$.   Then there exists constant $M\ge 0$ such that
\begin{align*}
0&\le &\frac{1}{t}\sum_{\tau=0}^t \Big( L(\Theta(\tau),\Lambda^*) - L(\Theta^*,\Lambda(\tau)) \Big) \leq   \frac{\Delta}{2\alpha t} + \frac{\alpha M}{2}
\end{align*}
where $(\Theta^*,\Lambda^*)$ is a solution to optimisation problem (P1), $\Delta=|| \Theta(0)-\Theta^*)||_2^2 +||\Lambda(0)-\Lambda^*||_2^2$ and $||\cdot ||_2$ denotes the usual Euclidean norm.   
\end{lemma}
\begin{proof}
\textcolor{black}{Optimisation} problem $(P1)$ is convex, the objective and constraint functions are differentiable and the Slater condition is satisfied. The result now follows by direct application of Lemma \ref{th:main} in the Appendix.
\end{proof}

Since $\Delta/(2\alpha t)\rightarrow 0$ as $t\rightarrow \infty$, Lemma \ref{th:converge} tells us that update (\ref{eq:dyn01})-(\ref{eq:dyn02}) converges to a ball around an optimum  $(\Theta^*,\Lambda^*)$, the size of the ball decreasing with step size $\alpha$.  The size of the ball is measured in terms of metric $L(\Theta,\Lambda^*) - L(\Theta^*,\Lambda)$, and recall that by complementary slackness $L(\Theta^*,\Lambda^*)=\sum_{u\in \U} U( \hat{R}_{u}(\Theta^*))$.

\subsection{Message Passing and Implementation}\label{sec:impl}
Algorithm \ref{algo1} can be implemented in a distributed manner.  Namely,  each base station $b\in \B$ carries out local tilt angle updates according to (\ref{eq:dyn01}) and (\ref{eq:dyn02}), and also carries out update  (\ref{eq:dyn01a}) for each user $u$ associated with base station $b$.    For this, each base station $b$ needs to evaluate (\ref{eq:subgrad01})-(\ref{eq:subgrad02}).   Evidently (\ref{eq:subgrad01a})-(\ref{eq:subgrad02}) can be evaluated using locally available information (the tilt angle of base station $b$ and the current downlink throughput of user $u$ associated with base station $b$).   In contrast, evaluating (\ref{eq:subgrad01}) requires information sharing between base stations.   Specifically, it is necessary to evaluate
\begin{align}
\label{eq:utility-gradiant}
\sum_{u \in \U}\frac{\partial \hat{r}_{u}}{\partial \theta_{b}}
&=\sum_{u\in\{u\in\U: b(u)=b\}}\frac{\partial \hat{r}_{u}}{\partial \theta_{b}}+\sum_{u\in\{u\in\U: b(u)\ne b\}}\frac{\partial \hat{r}_{u}}{\partial \theta_{b}}
\end{align}
The first term in (\ref{eq:utility-gradiant}) is the sensitivity of the throughput of users associated to base station $b$ to changes in its tilt angle $\theta_b$.   This can either be directly measured by base station $b$ (by perturbing the tilt angle), or calculated using 
\begin{align}
\frac{\partial \hat{r}_{u}}{\partial \theta_{b}}
&=
\frac{w}{N_{sc}}\sum_{n=1}^{N_{sc}}  \frac{\partial G_{v}(\theta_{b},d_{b,u})}{\partial \theta_{b}}
\end{align}
where
\begin{IEEEeqnarray}{c}
 \frac{\partial G_{v}(\theta_{b},d_{b,u})}{\partial \theta_{b}}=\frac{2.4 \log10 }{ \theta_{3dB^{2}}}(\theta_{b(u),u}-\theta_{b})
\end{IEEEeqnarray}
This calculation requires knowledge of the pointing angle $\theta_{b(u),u}$ between base station $b$ and user $u$.   This pointing angle can be determined from knowledge of the location of users, information which is usually available to modern base stations since Location Based Services (LBS) are of high importance for mobile network providers. For example, in the US carriers are required by FCC to provide location-based information of the mobile users for E911 services and to within a specified accuracy \cite{E911}.  Within Release 9 of 3GPP a set of enhanced positioning  methods are standardized for LTE\cite{3GPPTS136.355V9}.

The second term in (\ref{eq:utility-gradiant}) is the sensitivity of the throughput of users associated to base stations other than $b$ to changes in tilt angle $\theta_b$.    This can be calculated as
\begin{align}
-\frac{1}{N_{sc}}\sum_{u\in\{u\in\U: b(u)\ne b\}}\sum_{n=1}^{N_{sc}}  \gamma_{u,n}(\Theta)\frac{\partial \hat{G}_{v}(\theta_{b},d_{b})}{\partial \theta_{b}} \frac{\hat{H}_{u}(\theta_{b}) } {H_u(\theta_{b(u)})}
\end{align}
This requires user received power $\hat{H}_{u}(\theta_{b})$ from base station $b$, user received power $H_u(\theta_{b(u)})$ from the base station to which it is associated, the user SINR $\gamma_{u,n}$ and the pointing angle $\theta_{b(u),u}$.   All of this information is available to the base station to which the user is associated (via user equipment received power and SINR reports), but not to neighbouring base stations and so must be communicated to them.   

We note that antenna tilt angle updates are likely to occur on a relatively long time-scale in practice.   Capturing hourly based traffic patterns of the mobile users may therefore also provide relatively reliable traffic distribution information, which might also be used. 

 \subsection{Example}
\label{subsec: PE-high SINR}
We illustrate the application of the foregoing high SINR analysis to the scenario shown in Fig ~\ref{fig:simulation scenario}.   We use a simple scenario here to help gain insight, with a more realistic setup considered in detail in Section \ref{sec:perf}.  The scenario consists of regularly spaced base stations each with three sector antennas.   The base station radio parameters are detailed in  Table \ref{tb:parameters} 
based on 3GPP standard \cite{3GPPTR36.814V9}
. The users are primarily located in two clusters, as indicated in Fig ~\ref{fig:simulation scenario}.   One cluster of 16 users is associated with the first sector of base station 1, and the other cluster of 16 users with the third sector of base station 2.  Clustering of users creates a challenging tilt angle assignment task since a poor choice of tilt angles will have a strong effect on network performance.  Additionally, two users are located close to the mid-point between these base stations.   Ensuring adequate coverage at cell edges is commonly an issue for network operators and so we expect a performance tradeoff between serving these edge users and serving users located in the clusters. 
For concreteness, we select utility function $U(z)=z$, so that optimisation problem (P1) corresponds to maximising the network sum-throughput, subject to every user obtaining a minimum throughput of $64kps$ and to physical constraints that the allowable tilt angles must lie in the interval $[5,20]$ degrees.
\begin{table}[ht]
\label{tb:parameters}
\caption{simulation parameters  }
\centering
\begin{tabular}{|l|l|l|}
\hline
\multicolumn{3}{|c|}{Simulation parameters} \\
\hline
\hline
\multirow{4}{*}{Site and Sector} & Inter-site distance & $500m$ \\
 & Number of sectors & $3$ \\
 & Antenna max gain  $\tilde{G}_0$ & $15dBi$ \\
 & Antenna height $h$ & $25m$\\
 & Vertical half power beamwidth $\theta_{3dB}$ & $10^\circ$ \\
 & Transceiver power $P_{b,n}$ & $46dBm$ \\ \hline \hline
\multirow{3}{*}{Channel} & Exponential path loss factor $\beta$ & $3.76$  \\
 & Fixed path loss factor $\rho_0$ & $0.0316$ \\
 & Bandwidth $w$  & $10MHz$ \\ 
& Number of sub-carriers $N_{sc}$ & $1$ \\
& UE noise power $\eta_{n,u}$ & $-94.97 dBm$ \\ \hline \hline 
\multirow{2}{*}{Optmisation} & Min tilt angle $\underline{\theta}$ &$5^\circ$ \\
& Max tilt angle $\bar{\theta}$ & $20^\circ$\\
& Min throughput $\underline{r}$ & $64kbps$ \\
& Max throughput $\bar{r}$ & $10Mbps$ \\
& Step size $\alpha$ & $0.05$ \\
\hline
\end{tabular}
\end{table}
\begin{figure}
\centering
\includegraphics[width=0.6\columnwidth]{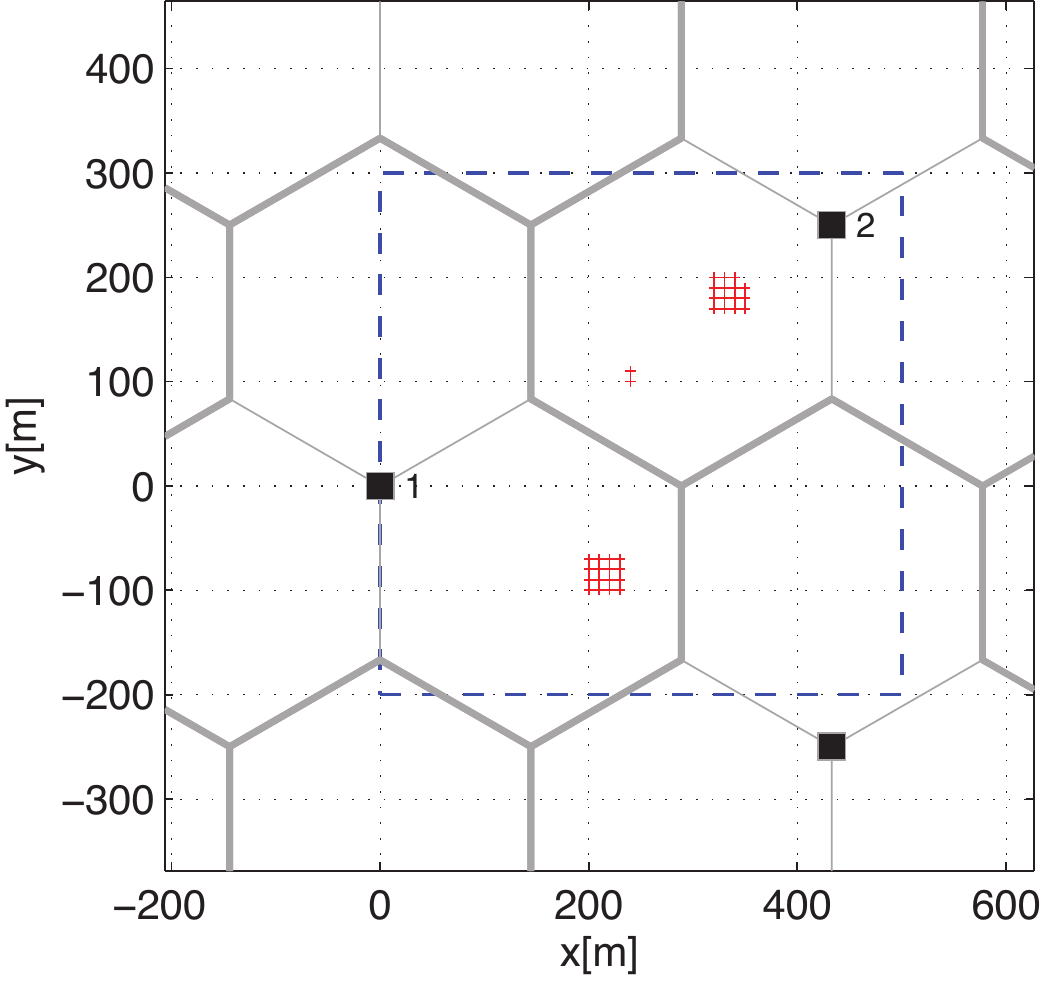}
\caption{Example network topology. Base-stations are indicated by solid squares labelled 1,2 and 3, UEs by dots.}
\label{fig:simulation scenario}
\end{figure}

Figure~\ref{fig:paremeters} shows tilt angle time histories for the two base stations when using Algorithm \ref{algo1}.   It can be seen that the tilt angles  converge to the optimum in less than 600 iterations.    Figure~\ref{fig:throughput} shows the corresponding network sum-throughput vs time.   Also shown in the network sum-throughput for fixed antenna angles of $8^\circ$.  Optimising the tilt angles increases the network sum throughput by almost factor of 
18 compared to the use of fixed angles.  As already noted,  the improvement is expected to be particularly pronounced in this simple example since the users are grouped into clusters, and so angling the antennas to point towards their respective clusters both greatly increases received power and decreases interference.    We note that significant performance gains are, however, also observed in the more realistic scenario studied in Section \ref{sec:perf} and this reflects the fundamental importance of antenna tilt angle to network performance.

\begin{figure}
\centering
\includegraphics[width=0.8\columnwidth]{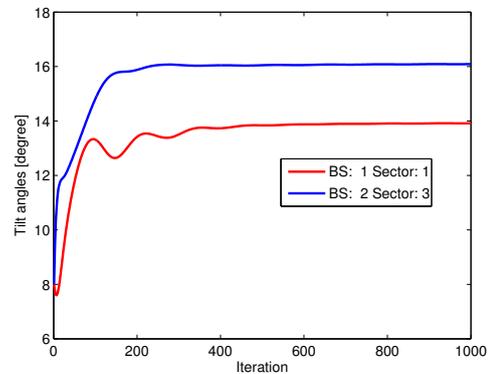}
\caption{ Tilt angle vs iteration number with Algorithm \ref{algo1}. }
\label{fig:paremeters}
\end{figure}
\begin{figure}
\centering
\includegraphics[width=0.8\columnwidth]{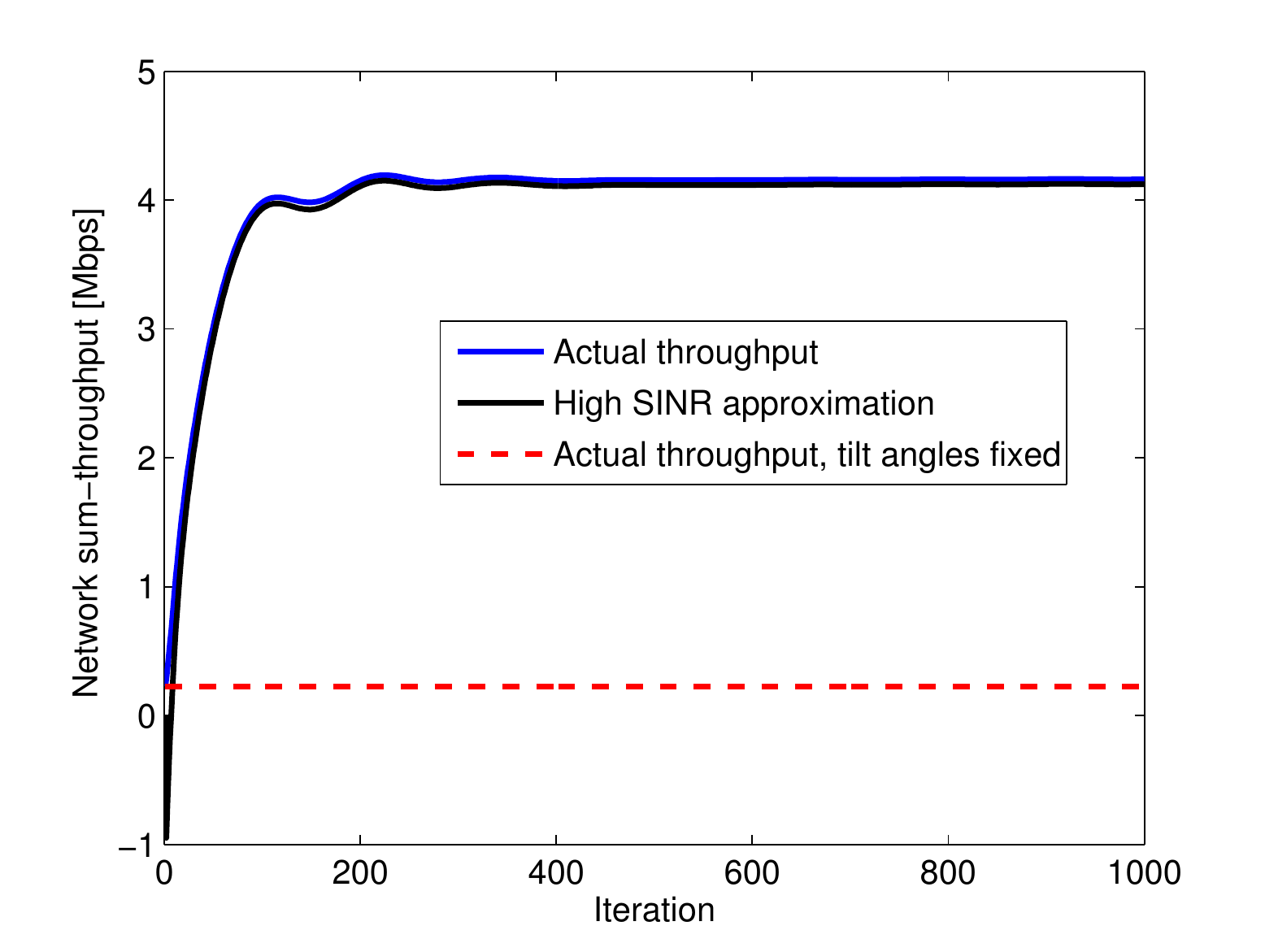}
\caption{ Normalised network sum-throughput vs iteration number with Algorithm \ref{algo1}.  Here, the network sum-throughput is normalised by dividing by the number of users in the network. 
}
\label{fig:throughput}
\end{figure}
\begin{figure}
\centering
\includegraphics[width=0.8\columnwidth]{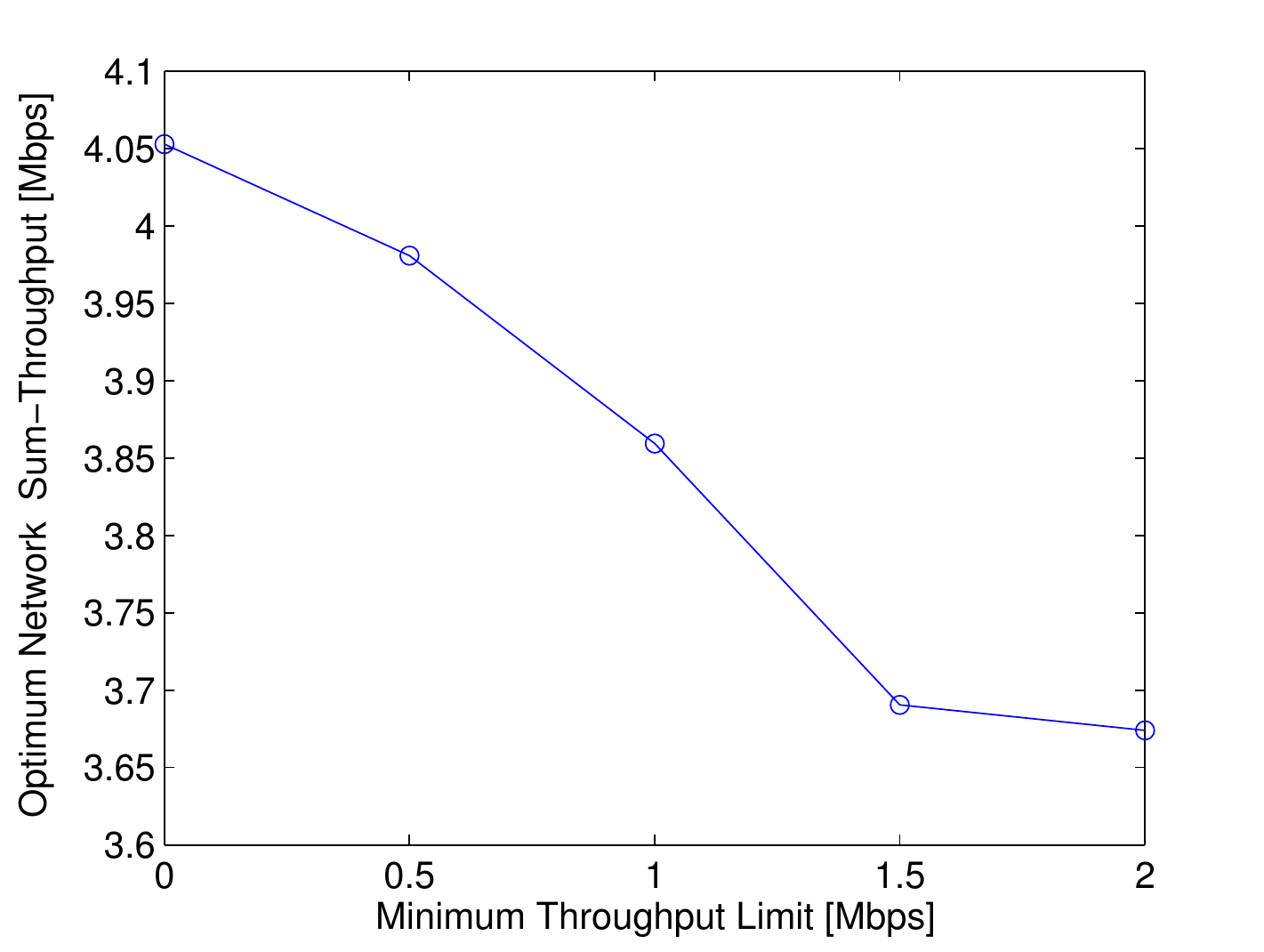}
\caption{ Normalised network sum-throughput as minimum throughput constraint $\underline{r}$ is varied.   The network sum-throughput is normalised by dividing by the number of users in the network. }
\label{fig:throughput-min}
\centering
\end{figure}

Fig \ref{fig:throughput-min} illustrates the impact of the minimum throughput constraint $\underline{r}$ on network sum-throughput.   It can be seen that as $\underline{r}$ is increased from zero to 2Mbps, the network sum-throughput decreases but that the impact is minor.   Note that as $\underline{r}$ is increased beyond 2Mbps the optimisation becomes infeasible as the stations at the cell edge are unable to support such high rates. 

As discussed in Section \ref{sec:impl}, UE location information is used when calculating  (\ref{eq:subgrad01}) in Algorithm \ref{algo1}.   In practice this location information will be approximate in nature.    Fig \ref{fig:location-error} plots the optimised network sum-throughput vs the standard deviation of the location error when zero-mean gaussian noise is added to the true user locations.    It can be seen that, as might be expected, the optimised sum-throughput falls as the noise level is increased.  However, the decrease is small (less than 5\%) even for relatively large location errors.

\begin{figure}
\centering
\includegraphics[width=0.8\columnwidth]{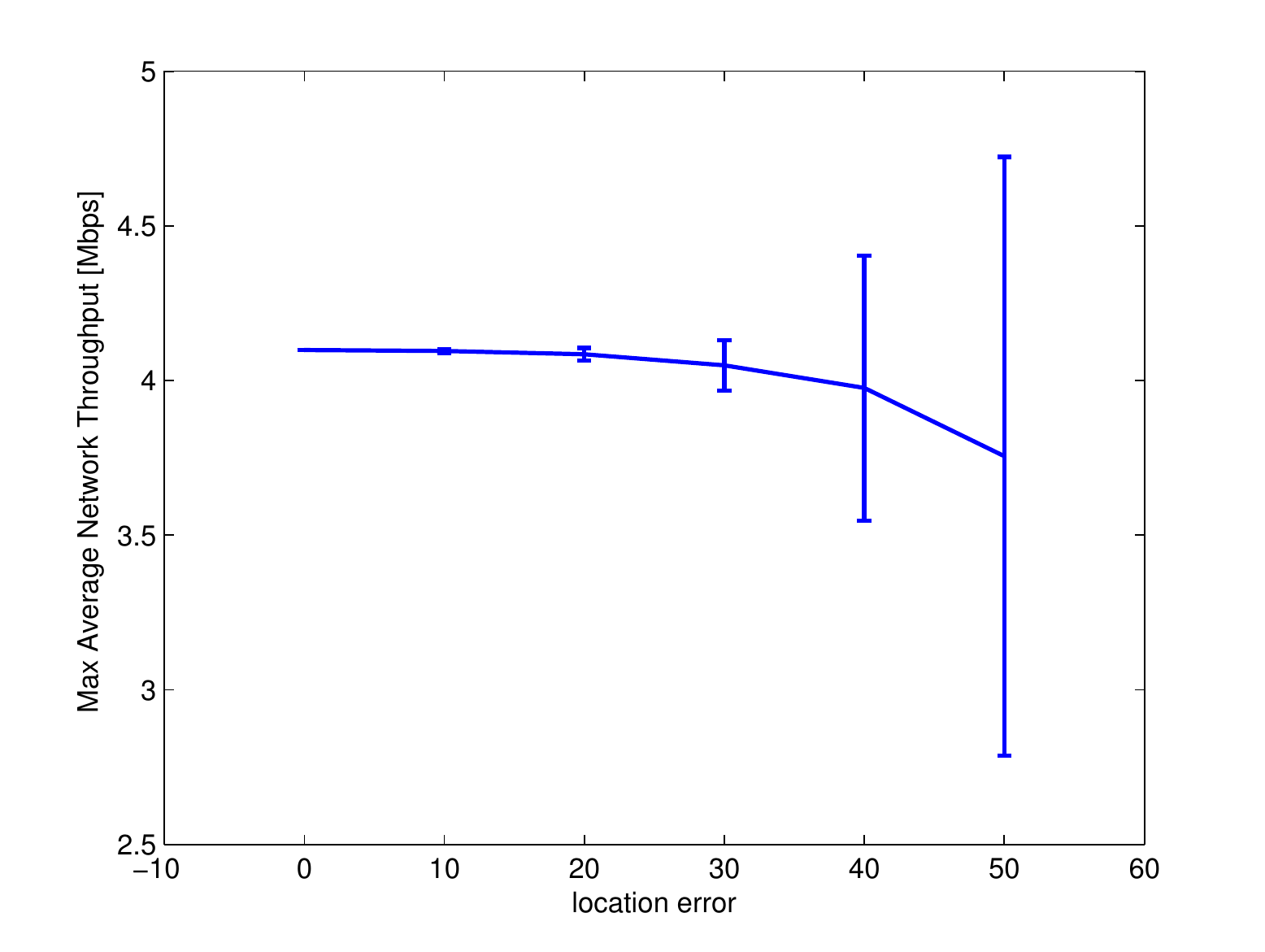}
\caption{ Optimised network sum-throughput vs. magnitude of location error, error bars indicate the standard deviation for \textcolor{black}{100} runs of simulation for each point.  The network sum-throughput is normalised by dividing by the number of users in the network.}
\label{fig:location-error}
\end{figure}
 \section{Any SINR: Proportional Fair Rate Allocation}
\label{sec:proportionalfair}

In this section we relax the assumption of operation in the high SINR regime.    However, this comes at the cost of restricting attention to proportional fair rate allocations.   We consider the following utility fair optimisation problem (P2):
\begin{align}
\label{eq:main-objective-log}
\max_{\Theta} & \sum_{u\in \U} \log  R_{u} (\Theta) \\ 
s.t.\quad & \underline{\theta} \le \theta_b \le \bar{\theta}, \quad b\in\B \label{eq:cons1a}\\
& \log\underline{r} \leq \log R_{u}(\Theta)  , \quad u \in \U \label{eq:cons1b}
\end{align}
where $R_u(\Theta)$ is given by (\ref{eq:capacity-truncated}).
\subsection{Convexity Properties}
We recall the following,
\begin{lemma}[\cite{papandriopoulos2008optimal}]{\label{th:four}} $h(x)=\log (\log(1+e^{x}))$ is concave and non-decreasing in $x\in\mathcal{R}$.
\end{lemma}
Turning now to $R_u(\Theta)$, we begin by observing that
\begin{lemma} {\label{th:five}} $\log(r_{u}(\Theta))$ is concave in $\Theta$.
\end{lemma}
\begin{proof}
From (\ref{eq:user-capacity-shannon}) we have
\begin{align}
\log(r_{u}(\Theta))&=\frac{w}{N_{sc}}\sum_{n=1}^{N_{sc}}\log(\log(1+\kappa\gamma_{u,n}(\Theta))\\
&=\frac{w}{N_{sc}}\sum_{n=1}^{N_{sc}}\log(\log(1+ e^{\hat{r}_{u,n}(\Theta)}))
\end{align}
where $\hat{r}_{u,n}(\Theta)=\log(\kappa\gamma_{u,n}(\Theta))$.  That is, the mapping from vector $\Theta$ to $\log(r(\Theta))$ is the vector composition of $h(x)$ in Lemma \ref{th:five} and $\hat{r}_{u,n}(\Theta)$.  By Lemmas  \ref{th:one} and \ref{th:two}, $\hat{r}_{u,n}(\Theta)$ is concave in $\Theta$.   By\textcolor{black}{~\cite[p86]{boyd2004convex}}, the vector composition a non-decreasing concave function and a concave function is concave.
\end{proof}
and
\begin{theorem}\label{lem:six}
$\log  R_{u} (\Theta)$ is concave in $\Theta$.
\begin{proof}
From (\ref{eq:capacity-truncated}) we have 
\begin{align}
\log R_{u}(\Theta)&= \log\min\{\bar{r}, r_{u}(\Theta)\} \\
&\stackrel{(a)}{=}\min\{ \log\bar{r}, \log r_{u}(\Theta)\}
\end{align}
where $(a)$ follows from the fact that the $\log$ function is monotonically increasing.  By Lemma \ref{th:five},  $\log r_{u}(\Theta)$ is concave.  Since the $\min$ function is concave non-decreasing, when composed with  $\log r_{u}(\Theta)$ it is concave i.e.  $\log \hat{R}_u(\Theta)$ is concave in $\Theta$.  
\end{proof}
\end{theorem}
\subsection{Convex Optimisation}
It follows from Theorem \ref{lem:six} that the objective of optimisation problem (P2) is concave.   Constraints (\ref{eq:cons1a}) are linear (so convex).   The RHS of constraint (\ref{eq:cons1b}) is concave, again by Theorem \ref{lem:six}, and so this constraint is convex.  It follows that optimisation problem (P2) is convex and a solution exists.  
\subsection{Distributed Algorithm}
The Slater condition is satisfied and strong duality holds.   We can therefore apply a similar approach as in Section \ref{sec:algo} to develop a distributed algorithm for finding the optimal antenna tilt angles.

The Lagrangian is:
\begin{IEEEeqnarray}{rCl}
L(\Theta,\Lambda) &=& -\sum_{u\in\U} \log R_u(\Theta)+\sum_{u\in\U}\lambda_u^1(\log \underline{r}-\log R_u(\Theta))  \nonumber \\
&& \: + \sum_{b\in\B}\lambda^2_b (\underline{\theta}-\theta_b) +  \sum_{b\in\B}\lambda^3_b (\theta_b-\bar{\theta})
\end{IEEEeqnarray}
We can now apply Algorithm \ref{algo1} to solve (P2) provided we use the appropriate gradients:
\begin{IEEEeqnarray}{rCl}
\partial_{\theta_b}L(\Theta,\Lambda) &=& -\sum_{u\in\U}\left(1 +\lambda_u^1\right)\partial_{\theta_b}  (\log{R}_u(\Theta)) - \lambda^2_b +  \lambda^3_b \label{eq:logsubgrad01} \IEEEeqnarraynumspace \\
\partial_{\lambda_u^1}L(\Theta,\Lambda)&=& \log \underline{r}-\log R_u(\Theta)\\
\partial_{\lambda_b^2}L(\Theta,\Lambda)&=& \underline{\theta} - \theta_b\\
\partial_{\lambda_b^3}L(\Theta,\Lambda)&=&\theta_b - \bar{\theta}
\label{eq:logsubgrad02}
\end{IEEEeqnarray}
with\newline
\begin{minipage}[c]{0.9\columnwidth}
\footnotesize
\begin{align*}
\frac{\partial r_{u}}{\partial \theta_{b}}&=
\begin{cases}
\frac{\partial G_{v}(\theta_{b},d_{b,u})}{\partial \theta_{b}}
\frac{\kappa H_{u}(\theta_{b(u)})}{\kappa H_{u}(\theta_{b(u)})+I_{u \in b(u)}}  
& b=b(u) \\ 
\frac{\partial \hat{G}_{v}(\theta_{b},d_{b,u})}{\partial \theta_{b}}
\bigg( 
\frac{\hat{H}_{u}(\theta_{b(u)})} {\kappa H_{u}(\theta_{c})+I_{u \in \B \setminus\{b(u)\} } }-\frac{\hat{H}_{u}(\theta_{b(\acute{u})})} {I_{u \in \B \setminus\{b(u)\}}}
 \bigg)  
 & b\ne b(u) \\ 
\end{cases}
\end{align*}
\end{minipage}
and
\begin{IEEEeqnarray}{c}
\label{eq:interference-and-noise-u}
I_{u \in b(u)} =\sum_{c \in \B\setminus\{b(u)\}}  \hat{H}_{u}(\theta_{c})+\eta_{u,n} 
\end{IEEEeqnarray}

 \subsection{Example}
We revisit the example in section~\ref{subsec: PE-high SINR}.   Fig \ref{fig:throughput_two_regimes} compares the results for optimisation problems $(P1)$ and $(P2)$.  Fig \ref{fig:throughput_two_regimes}(a) shows the sum-throughput from Fig \ref{fig:throughput} and when solving proportional fair allocation problem $(P2)$.  As expected, the sum throughput is lower for the proportional fair allocation.  Fig \ref{fig:throughput_two_regimes}(b) compares the sum-log-throughput.  As expected, the sum-log-throughput is higher for the proportional fair allocation problem $(P2)$.   Fig \ref{fig:throughput_dis_two_regimes} shows detail of the throughputs assigned to individual users to maximise sum-throughput and for proportional \textcolor{black}{fairness}.   It can be seen that the throughput assignments are broadly similar in both cases, with the primary difference being the throughputs assigned to the two users located at the cell edge (numbered 33 and 34 in Fig \ref{fig:throughput_dis_two_regimes}). The proportional fair allocation assigns significantly higher rate to these edge stations than does the max sum-throughput allocation.
\begin{figure}
\centering
\includegraphics[width=1\columnwidth]{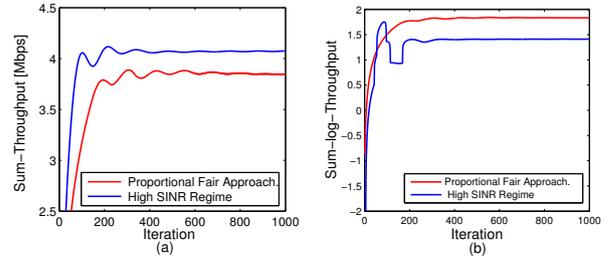}
\caption{(a) Comparing normalised sum-throughput from Fig \ref{fig:throughput} and when solving proportional fair allocation problem $(P2)$ -- as expected, the sum throughput is lower for the proportional fair allocation.  (b) Comparing normalised sum-log-throughput -- as expected, the sum-log-throughput is higher for the proportional fair allocation.  In all cases the network throughput is normalised by dividing by the number of users in the network. }
\label{fig:throughput_two_regimes}
\end{figure}
\begin{figure}
\centering
\includegraphics[width=0.9\columnwidth]{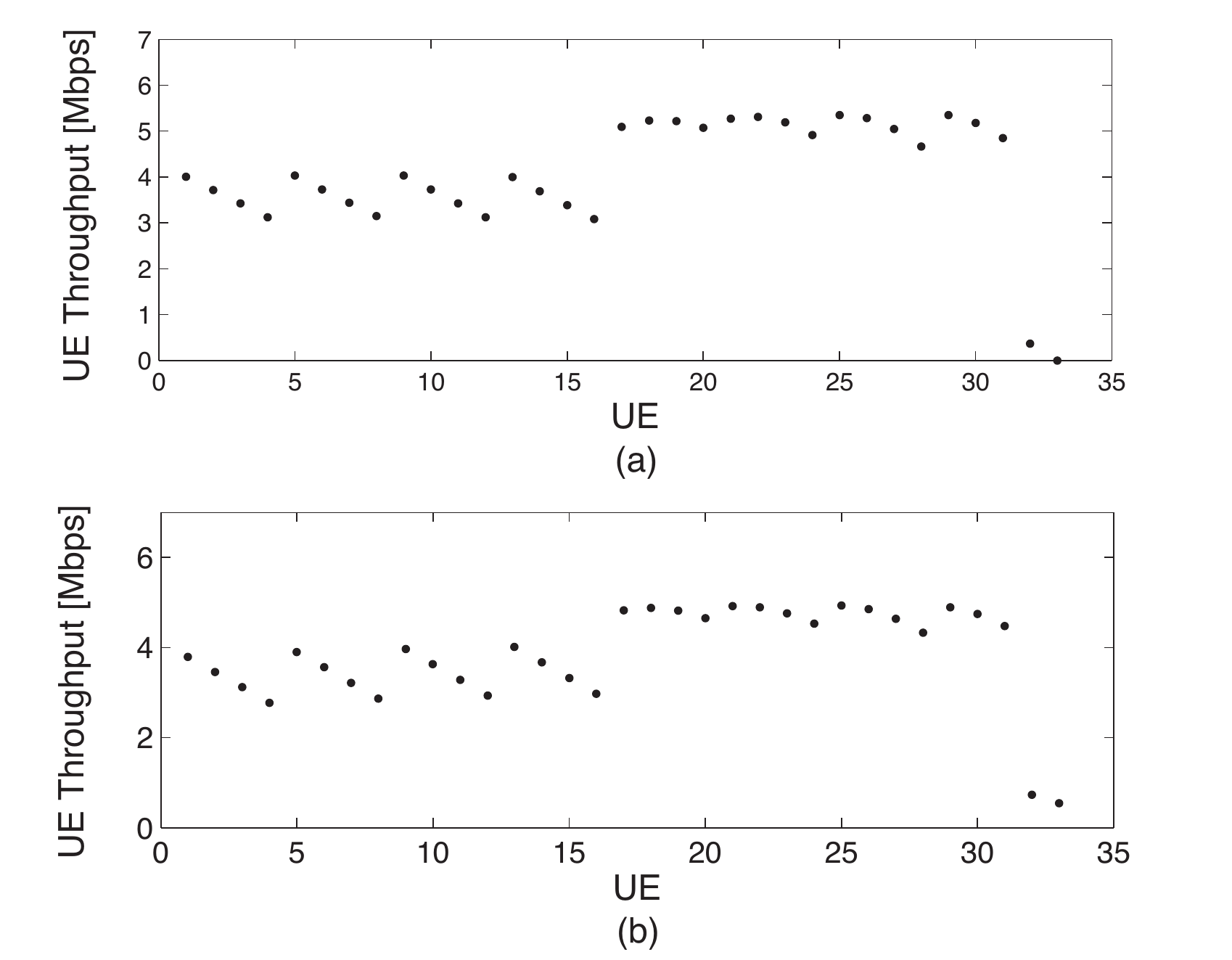}
\caption{User throughput assignments for (a) sum-throughput maximisation in the high SINR regime and (b) proportional fair rate allocation.}
\label{fig:throughput_dis_two_regimes}
\end{figure}

\section{Performance Evaluation}\label{sec:perf}
\begin{figure}
\centering
\hfil
\subfloat[Street map]{
\includegraphics[width=1.2in]{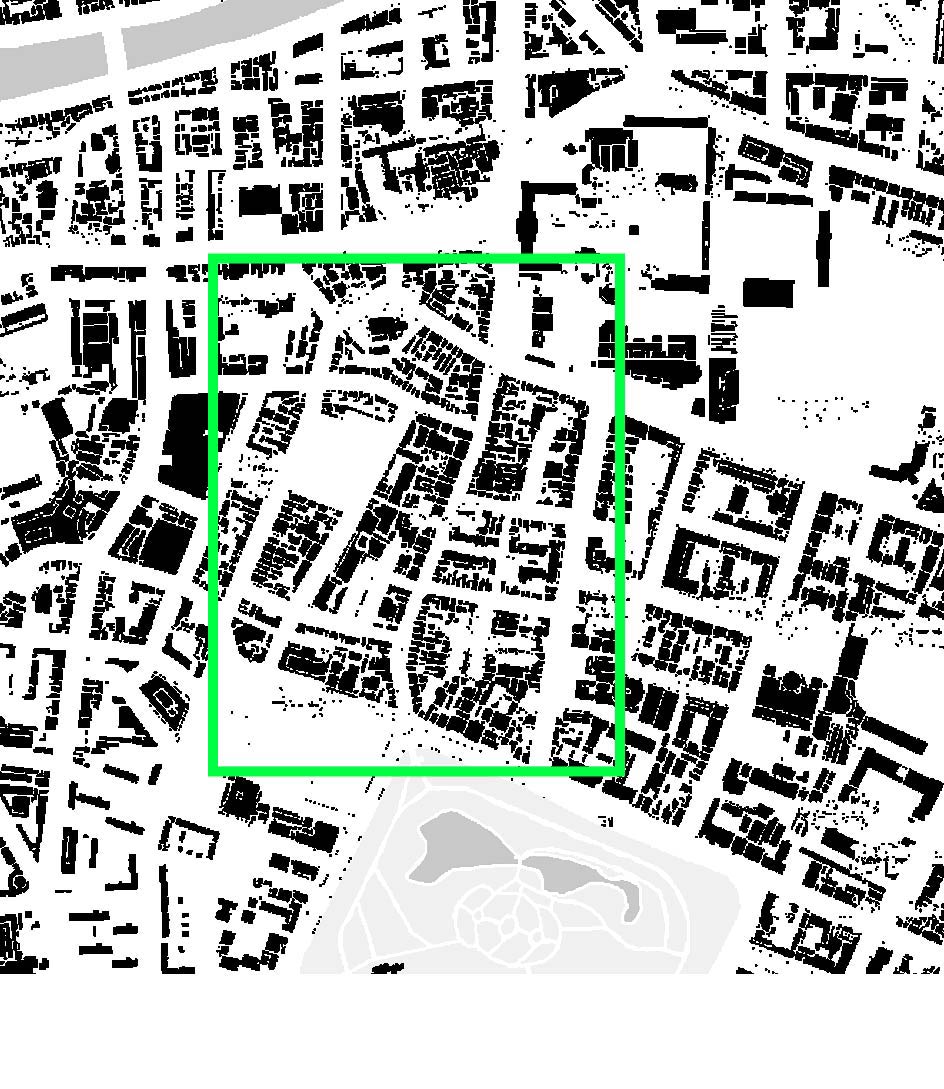}
}
\label{fig:dublin_building}
\subfloat[User and base-station locations. ]{\includegraphics[width=2in]{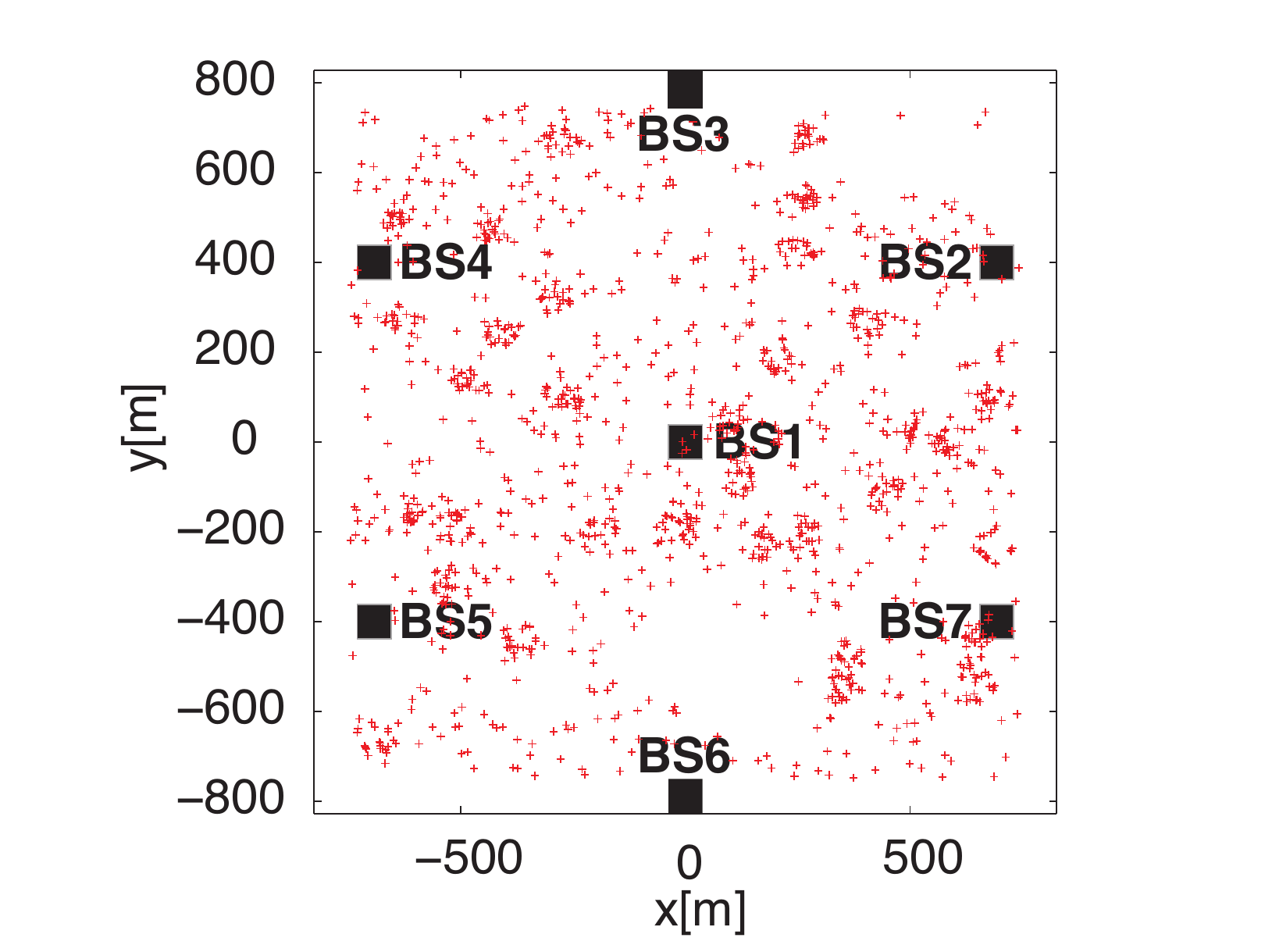}
\label{fig:user_dist}}
\caption{Dublin, Ireland example.}
\label{fig:dublin_sim}
\end{figure}
In this section we consider a realistic example based on data from the cellular network covering Grafton Street and Dawson Street in downtown Dublin, Ireland, see Fig \ref{fig:dublin_building}.   These are major shopping streets close to the centre of Dublin city, with a large number of cellular users.    We consider a section of the network with 21 sectors in  a $1500 m \times 1500m $ area and with an inter-site distance of $800m$.  Environmental characteristics are derived from experimental measurement data with a combination of non-line of sight and line-of sight paths.   Path loss and log-normal shadow fading parameters are derived from \cite{3GPPTR36.814V9} for macro urban scenarios and detailed Table \ref{tb:parameters_dublin}.   There are 1350 users, with locations as shown in Fig \ref{fig:user_dist}.  We focus on the performance experienced by the 388 users associated with the centre base station (indicated by BS1 in Fig \ref{fig:user_dist}).
\begin{table}[ht]
\label{tb:parameters_dublin}
\caption{Dublin scenario simulation parameters}
\centering
\begin{tabular}{|l|l|l|}
\hline
\multicolumn{3}{|c|}{Dublin Scenario Simulation Parameters} \\
\hline
\hline
\multirow{1}{*}{Site and Sector} & Inter-site distance & $800m$ \\
\hline
\hline
\multirow{3}{*}{Channel} & NLOS exponential path loss factor & $3.9$ \\
& NLOS fixed path loss factor  & $10^{-2.1}$ \\
& LOS exponential path loss factor & $2.2$ \\
& LOS fixed path loss factor & $10^{-3.4}$ \\
& Shadow fading standard deviation  & $6$ \\
& Shadow fading mean  & $0$ \\
\hline \hline 
\multirow{1}{*}{Optimisation} & Step size $\alpha$ & $0.01$ \\
\hline
\end{tabular}
\end{table}
Figures \ref{fig:sim} shows the proportional fair rate allocation.  For comparison, results are also shown when a fixed tilt angle of $8^\circ$ is used.    It can be seen from Fig \ref{fig:sim}(a) that the sum-log-throughput objective function is improved by 22\% by tilt angle optimisation, and that Algorithm \ref{algo1} converges rapidly to the optimal allocation.   From the cumulative distribution fucntion (CDF) in Fig \ref{fig:sim}(b) it can be seen the user throughputs are also significantly increased, with the median throughput increased by almost a factor of 4 compared to use of fixed angles.
\begin{figure}
\subfloat[Normalised Sum-log-throughput]{
\includegraphics[width=0.5\columnwidth]{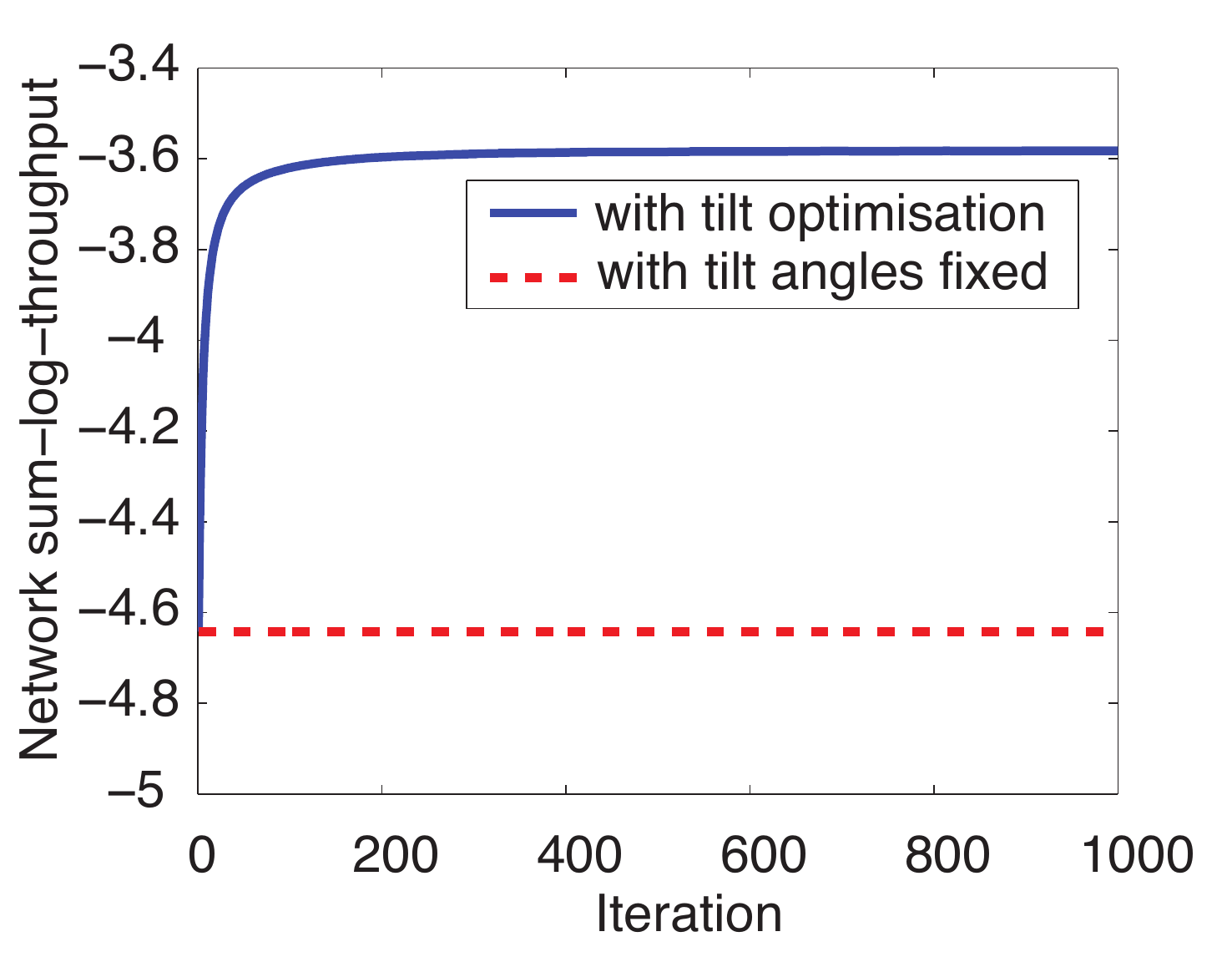}\label{fig:dublin_throughput_log}
}
\subfloat[CDF of user throughputs.]{
\includegraphics[width=0.5\columnwidth]{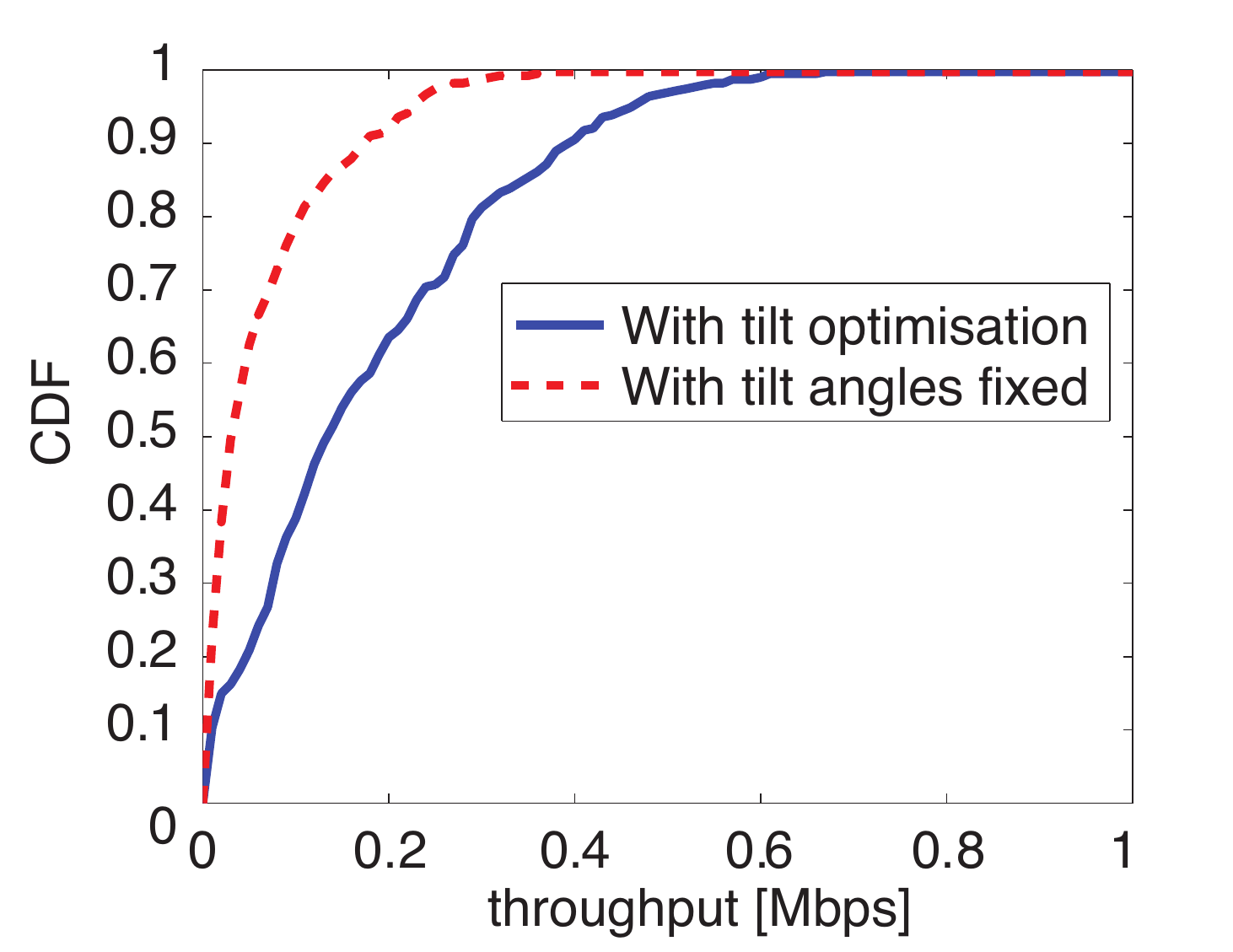}\label{fig:dublin_cdf}
}
\caption{Proportional fair rate allocation, Dublin example.  For comparison, results are also shown when a fixed tilt angle of $8^\circ$ is used (indicated by dashed lines). }
\label{fig:sim}
\end{figure}
%
\color{black} \section{LTE SIMO Links and MMSE Post-Processing}
In this section we extend the performance evaluation to consider LTE SIMO links with one transmit antenna on the BS and two receive antennas at the UE.    The presence of two antennas at the receiver allows the UE to cancel one interferer.  Hence, if interference is dominated by a single transmitter then we expect the use of SIMO links will allow inter-cell interference to be significantly reduced.   Our interest here is in the impact that this interference cancellation has on the size of throughput gain achievable by tilt angle adjustment.

We consider a SIMO link with linear Minimum Mean Square Error (LMMSE) post processing applied to the received signal to mitigate neighbouring cell interference.  Defining channel vector  $\textbf k_{u}=[k_{1}\ k_{2}]^T$, the channel gain for user $u$ is:
\begin{IEEEeqnarray}{c}\label{eq:v-mat} 
\textbf{k}_{u}=\sqrt{\frac{H_{u}(\theta_{b})}{P}}\sqrt{10^{S_{b,u}/10}}\textbf{q}_{b,u}.
\end{IEEEeqnarray}
where $S$ is a zero mean Gaussian random variable representing slow fading effects, $\textbf{q}_{b,u}$ is a Rayleigh flat fading vector and $P$ is the power of the transmitted signal assuming all base stations transmit at $P=P_{b,u}$.   We can consider the elements of  $\textbf{q}_{b,u}$ to be independent complex random Gaussian processes corresponding to the channels of base station $b$ and user $u$, provided that the antenna elements are sufficiently separated (typically on  the order of half a wavelength apart).  We identify the inter-cell interference vector  $\textbf{v}_{u}=[v_{1}\ v_{2}]^T$ for user $u$ by the strongest interferer:
\begin{IEEEeqnarray}{c}\label{eq:max_interference}
\textbf{v}_{u}=\sqrt{\max_{c \in \B \backslash \{b(u)\}} \{\hat{H}_{u}(\theta_{c}) 10^{S_{c,u}/10}\}}\textbf{q}_{c_{max},u}
\end{IEEEeqnarray}
The remaining inter-cell interference is modelled as spatially white Gaussian noise\cite{ETSITR 125996V11}, which comprises the noise vector $\boldsymbol{n}_{u}=\begin{bmatrix}
n_{1} \\
n_{2}\\
\end{bmatrix}$ where $n_{1}$ and $n_{2}$ are independent Gaussian variables:
\begin{align}\label{eq:rest_interference}
N_{0}=E[n_{1}n_{1}^{H}]=E[n_{2}n_{2}^{H}]=\sum_{\shortstack{$ \scriptstyle c \in \B \backslash \{b(u)\}$ \\ $ \scriptstyle c \neq c_{max}$}} \hat{H}_{u}(\theta_{c}) +\eta_{u}
\end{align}
Hence, the received signal (${\boldsymbol y}$) is given by:
\begin{IEEEeqnarray}{c}\label{eq:received-MMSE}
\textbf{y}_{u}= \textbf{k}_{u}\textbf{x}+\textbf{v}_{u}+\textbf{n}_{u}
\end{IEEEeqnarray}
with $E[xx^{H}]=P$.  The linear MMSE combining vector $\textbf{w}_{u}=[w_1\ w_2]^T$, is given by:
\begin{IEEEeqnarray}{c}\label{eq:MMSE-mat}
\textbf{w}_{u}=\textbf{k}_{u}^{H}{(\textbf{k}_{u} \textbf{k}_{u}^{H}+\frac{\boldsymbol{\Phi}+N_{0}\textbf{I}}{P})}^{-1}
\end{IEEEeqnarray}
where{\boldmath $\Phi$} is the autocorrelation of interference vector $\textbf{v}$:
\begin{IEEEeqnarray}{c}\label{eq:v-mat}
\boldsymbol{ \Phi}=E[\textbf{v}_{u}\textbf{v}_{u}^{H}]= 
\begin{bmatrix}
|v_{1}|^{2} & v_{1}v_{2}^{H}\\
v_{2}v_{1}^{H} & |v_{2}|^{2}\\
\end{bmatrix}=
\begin{bmatrix}
\phi _{11} & \phi _{12}\\
\phi _{21} & \phi _{22}\\
\end{bmatrix}
\end{IEEEeqnarray}
By applying the MMSE weights on the received signal, the post processing SINR is calculated as:
\begin{align}\label{eq:SINR-MMSE}
&\gamma_{u} ^{MMSE}= \nonumber \\
&\frac{P |w_{1}k_{1}+w_{2}k_{2}|^{2} }{{|w_{1}|}^2\phi _{11}+{|w_{2}|}^2\phi _{22}+ 2 Re\{w_{1}w_{2}^{H}\phi _{12}\}+N_{0}(|w_{1}|^{2}+|w_{2}|^2)}
\end{align} 
\begin{figure}
\subfloat[with fast fading]{
\includegraphics[width=0.5\columnwidth]{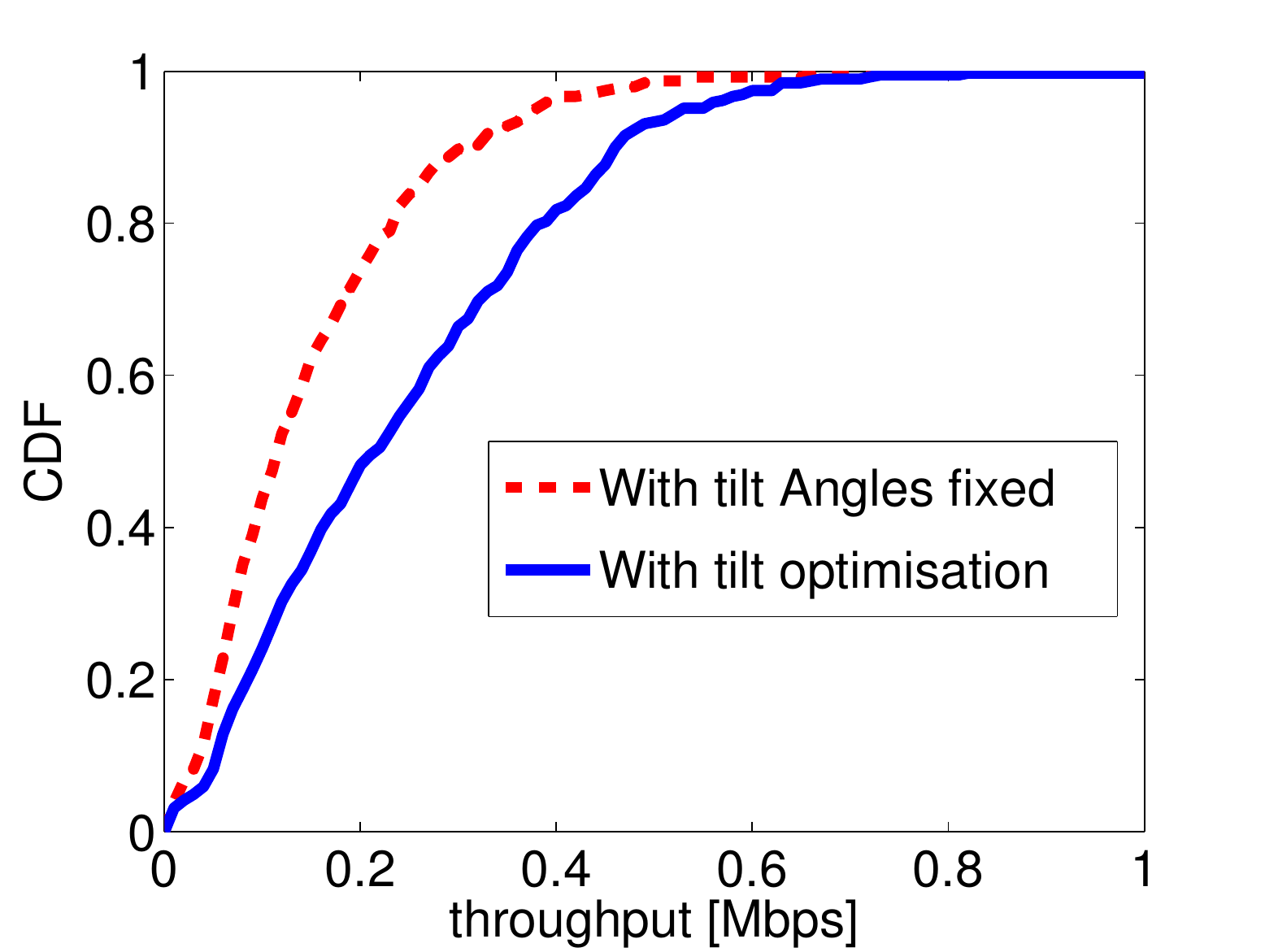}\label{fig:cdf_dublin_MMSE_ff}
}
\subfloat[without fast fading]{
\includegraphics[width=0.5\columnwidth]{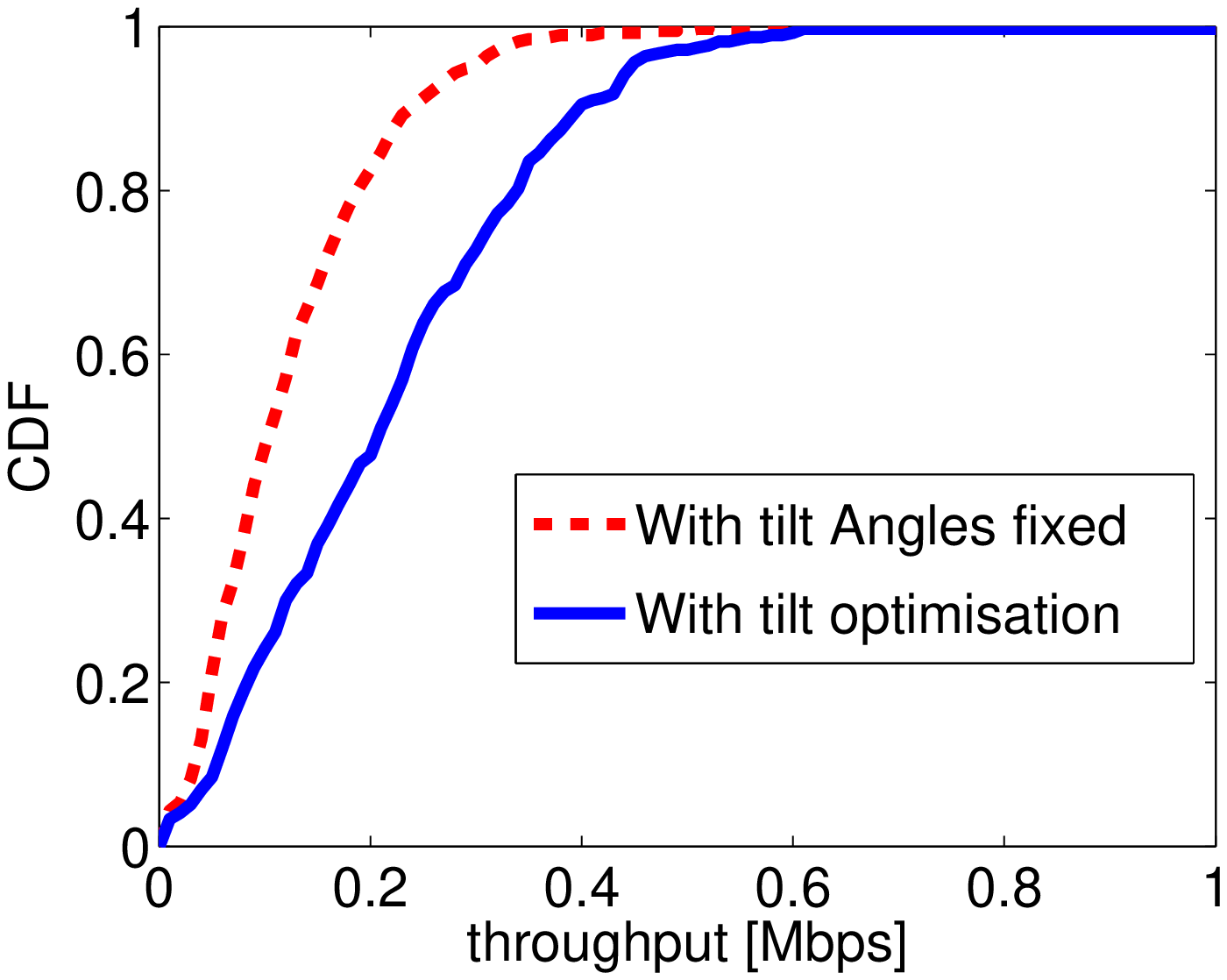}\label{fig:cdf_dublin_MMSE_nff}
}
\caption{User throughput CDFs with SIMO links and MMSE detection, Dublin example. Fast fading is modelled by generating $300$ samples using the 3GPP typical urban channel model, where the speed of the mobile user and carrier frequency are $3km/h$ and $2GHz$ respectively.}
\label{fig:MMSE_cdf}
\end{figure}
We average the post processing SINRs over the multipath fading samples. Using the averaged post processing SINRs,  user throughputs with and without tilt optimisation can be calculated using (\ref{eq:capacity-truncated}). 

Fig \ref{fig:MMSE_cdf} shows CDF of the user throughputs for SIMO links with MMSE detection, with and without flat fading.  As expected, the use of MMSE detection yields significant improvements in the user throughputs.  The throughput gains achieved by tilt optimisation can be compared for SISO links and for SIMO links with an optimal LMMSE detector by comparing Figs \ref{fig:dublin_cdf} and \ref{fig:MMSE_cdf}.   The gain in the mean user throughput  achieved by tilt optimisation is decreased from $83.07\%$ to $67.42\%$ when MMSE detection is employed.  However, the gain in the log-sum-rate (which is the objective function of optimisation ${P2}$) only changes from $22.29\%$ to $22.00\%$. That is, while MMSE detection enhances inter-cell interference mitigation, tilt optimisation can still yield significant improvements in network capacity.   

We can investigate this behaviour in more detail as follows.  Let
\begin{IEEEeqnarray}{c}\label{eq:epsilon}
\epsilon_{u}=\frac{\max_{c \in \B \backslash \{b(u)\}} \{\hat{H}_{u}(\theta_{c}) 10^{S_{c,u}/10}\}}{\sum_{c \in \B \backslash \{b(u)\} } \hat{H}_{u}(\theta_{c})10^{S_{c,u}/10}}
\end{IEEEeqnarray}
be the ratio of the largest interferer to the total interference experienced by a user $u$.   The CDF of $\epsilon$ for the Dublin example is shown in Fig \ref{fig:Interference_prof}(a).     It can be seen that approximately 40\% of users have $\epsilon$ values less than 0.5 i.e. for 40\% of users the the strongest interferer power is less than the sum of the power of the other interferers.   Fig \ref{fig:Interference_prof}(b) shows the corresponding spatial distribution of  $\epsilon$.  It can be seen that the strongest interferer is dominant at the edge of antenna sectors and along the nulls of the sector antennas. However, the intensity of the strongest interferer decreases along the edges of the base station coverage area and alongside the antennas.    Table \ref{tb:gains} details the throughput gains achieved by tilt angle optimisation for both SISO and SIMO links and for users with different $\epsilon $ ratios.  It can be seen that the throughput gain obtained by tilt angle optimisation for users with $\epsilon > .5$  is reduced when SIMO links are used.    However,  the gain is similar for both SISO and SIMO links for users with $\epsilon \leq .5$, once MMSE post processing is applied, and as noted above this consists of approximately 40\% of users. 

In summary, although the mean user throughput is improved for both fixed and optimal tilt angles for SIMO links with MMSE, tilt optimisation can still yield considerable performance gains.
\begin{figure}
\subfloat[]{
\includegraphics[width=.5\columnwidth]{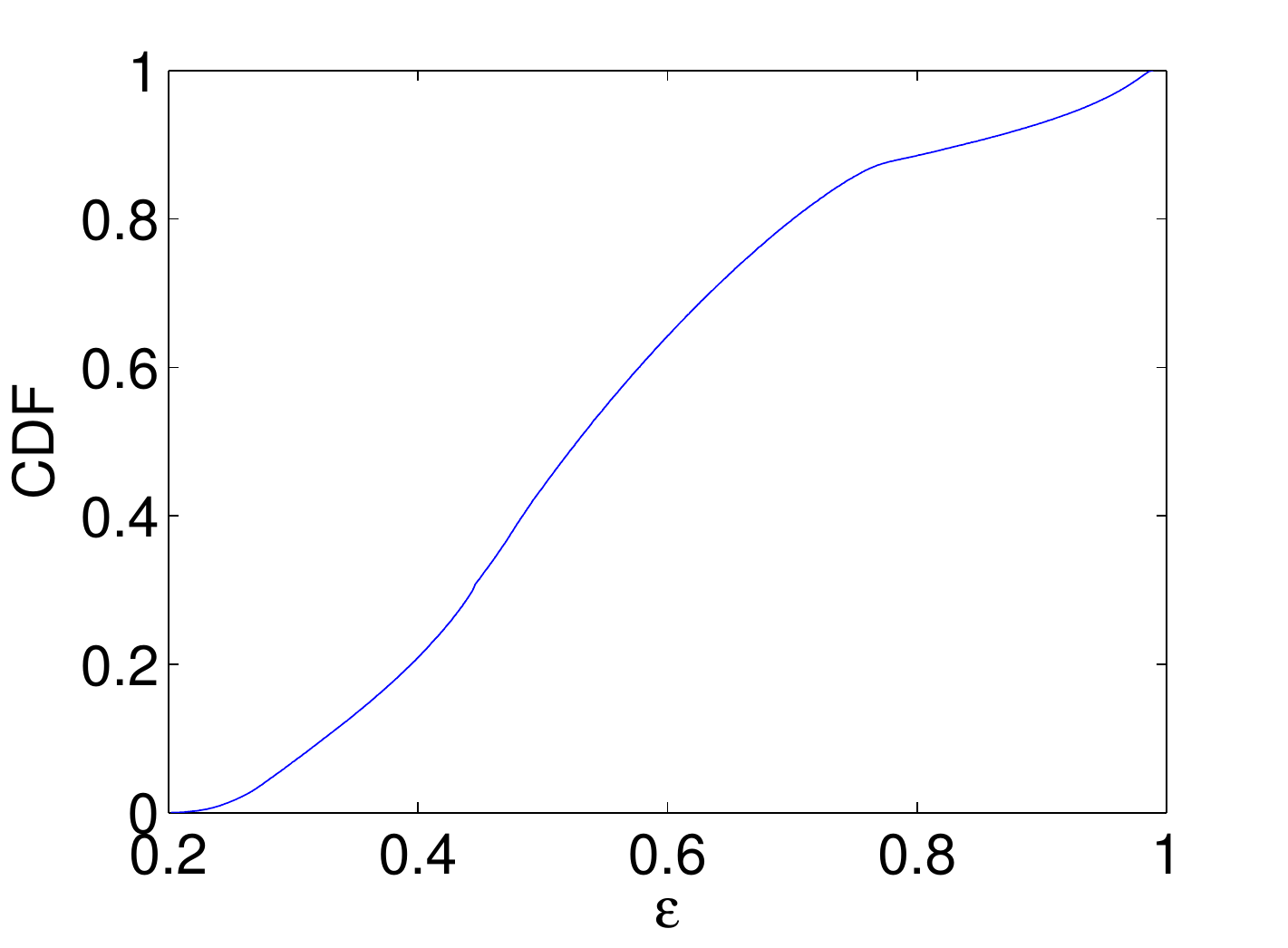}\label{interference_cdf}
}
\subfloat[]{
\includegraphics[width=.5\columnwidth]{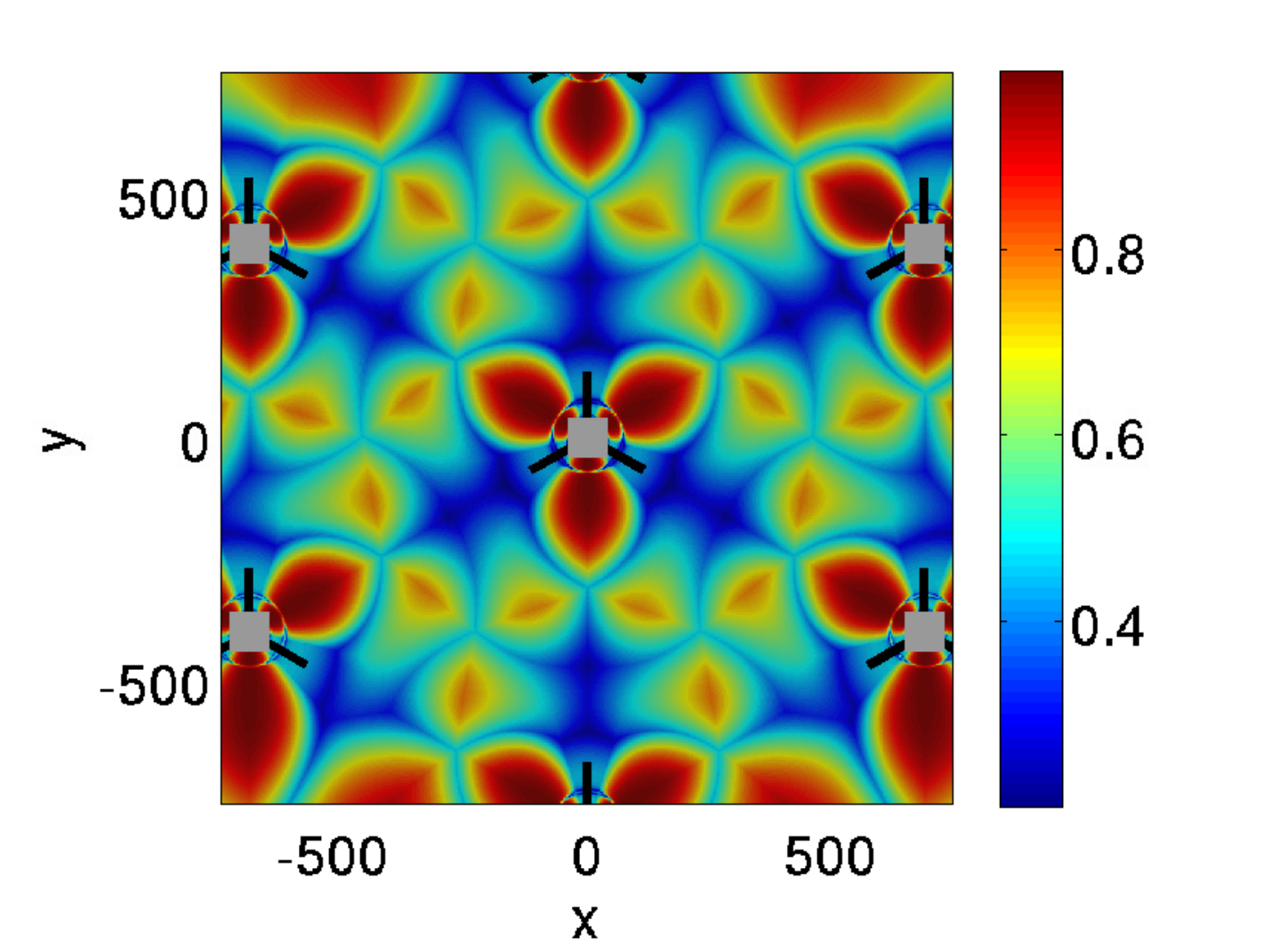}\label{interference_loc}
}
\caption{Contribution of the strongest interference to the total Interference:  (a) cumulative distributed function of $\epsilon$ in the central BS coverage area.    (b) Distribution of $\epsilon$ relative to the user positions. }
\label{fig:Interference_prof}
\end{figure}
\begin{table}
\caption{User throughput gains due to tilt angle optimisation for both SISO and SIMO links and vs $\epsilon$,  Dublin example.}
\centering
\begin{tabular}{|l||l|l|}
\hline
&\multicolumn{2}{l|}{Mean Throughput Gain [\%]}\\
\cline{2-3}
&\multicolumn{1}{l|}{SISO}&\multicolumn{1}{l|}{SIMO \& MMSE Detection}\\
\hline\hline
$ 0\le\epsilon_{u} \le .5$ &$54.5$&$52.3$\\
$ .5 < \epsilon_{u} \le .8$ &$128$&$85.99$\\
$ .8 < \epsilon_{u}  \le 1$ &$147$&$91.99$\\
\hline
$  0\le\epsilon_{u} \leq 1$ &$83.07$&$67.42$\\
\hline
\end{tabular}
\label{tb:gains}
\end{table}
\color{black}\section{Conclusions and Future Work}
\label{sec:conclusion}
\subsection{Conclusions}
In this paper we formulate adaptation of antenna tilt angle as a utility fair optimisation task.   Namely, the objective is to jointly adjust antenna tilt angles within the cellular network so as to maximise user utility, subject to network constraints.   Adjustments at base stations must be carried out jointly in a coordinated manner in order to manage interference.    This optimisation problem is non-convex, but we show that under certain conditions it can be reformulated as a convex optimisation.     Specifically, we show that (i) in the high SINR operating regime and with an appropriate choice of variables the optimisation is convex for any concave utility function, and (ii) in any SINR regime the optimisation can be formulated in a convex manner when the objective is a proportional fair rate allocation.    Since the optimisation is not well-suited to use of standard dual methods, we develop a primal-dual method for finding the optimal antenna tilt angles.  This approach is lightweight, making use of measurements which are already available at base stations, and suited to distributed implementation.   The effectiveness of the proposed approach is demonstrated using a number of simulation examples, including a realistic example based on the cellular network in Dublin, Ireland, and is found to yield considerable performance gains.


\section*{Acknowledgment}
The authors would like to thank Dr. Holger Claussen and Dr. David Lopez-Lopez of Bell Labs, Dublin for their valuable comments and suggestions.
\bibliography{references}
\bibliographystyle{ieeetr}

\section*{Appendix: Optimisation Problem}
Consider the optimisation problem
\begin{IEEEeqnarray*}{rCl}
&\min_{\mathbf{x}} f(\mathbf{x})\qquad s.t. \quad & g_i(\mathbf{x}) \le 0,\quad i=1,\cdots,m
\end{IEEEeqnarray*}
with $\mathbf{x}\in \mathcal{R}^n$, $f(\mathbf{x}):\mathcal{R}^n \rightarrow \mathcal{R}$ convex, $g_i(\mathbf{x}):\mathcal{R}^n\rightarrow \mathcal{R}$, $i=1,\cdots,m$ convex.  For simplicity we will assume that $f(\mathbf{x})$, $g_i(\mathbf{x})$ are differentiable, but this could be relaxed.  The optimisation problem is convex and so at least one solution exists, let $X^*$ denote the set of solutions   Assuming the Slater condition is satisfied, then strong duality holds and the KKT conditions are necessary and sufficient for optimality.   The Lagrangian is
\begin{IEEEeqnarray*}{rCl}
L(\mathbf{x},\mathbf{u}) = f(\mathbf{x}) + \sum_{i=1}^m u_i g_i(\mathbf{x})
\end{IEEEeqnarray*}
where $u_i$ is the multiplier associated with constraint $g_i(\mathbf{x})\ge 0$ and $\mathbf{u}=(u_1,\cdots,u_m)$.  At an optimum $\mathbf{x}^*\in X^*$, the multipliers must lie in set $U(\mathbf{x}^*)=\{\mathbf{u}: \mathbf{u}= \arg \sup_{\mathbf{u}\ge 0} L(\mathbf{x}^*,\mathbf{u})\}$.  

 \subsection{Gradient Algorithm}
We consider the following primal-dual update
\begin{IEEEeqnarray}{rCl}
x_j(t+1) &=& x_j(t) - \alpha \partial_{x_j}L(\mathbf{x}(t),\mathbf{u}(t)),\quad j=1,\cdots,n \label{eq:dyn1} \IEEEeqnarraynumspace\\
u_i(t+1) &=& \left[u_i(t) + \alpha \partial_{u_i}L(\mathbf{x}(t),\mathbf{u}(t))\right]^+,\quad i=1,\cdots,m \label{eq:dyn2}
\end{IEEEeqnarray}
where step size $\alpha>0$ and $\partial_{x_j}L(\mathbf{x},\mathbf{u})$, $\partial_{u_i}L(\mathbf{x}(t),\mathbf{u}(t))$ are subgradients of $L(\mathbf{x}(t),\mathbf{u}(t))$ with respect to $x_j$ and $u_i$ respectively.   We have $\partial_{x_j}L(\mathbf{x},\mathbf{u})=\partial_{x_j} f(\mathbf{x}) + \sum_{i=1}^m u_i \partial_{x_j} g_i(\mathbf{x})$ and $\partial_{u_i}L(\mathbf{x}(t),\mathbf{u}(t))=g_i(\mathbf{x})$ with $\partial_{x_j} f(\mathbf{x})$ a subgradient of $f(\mathbf{x})$ with respect to $x_j$, $\partial_{x_j} g_i(\mathbf{x})$ a subgradient of $g_i(\mathbf{x})$ with respect to $x_j$.  Projection $[z]^+=z$ when $z\ge0$, $0$ otherwise.  

 \subsection{Fixed Points}

\begin{lemma}[Fixed points]\label{th:fixed}
$(\mathbf{x}^*,\mathbf{u}^*)$ with $\mathbf{x}^*\in X^*$, $\mathbf{u}^*\in U(\mathbf{x}^*)$ is a fixed point of the dynamics (\ref{eq:dyn1})-(\ref{eq:dyn2}).
\end{lemma}
\begin{proof}
From the KKT conditions, $\partial_{x_j}L(\mathbf{x}^*,\mathbf{u}^*)=0$ and so $(\mathbf{x}^*,\mathbf{u}^*)$ is a fixed point of (\ref{eq:dyn1}).  Since $(\mathbf{x}^*,\mathbf{u}^*)$ is feasible, $\partial_{u_i}L(\mathbf{x}^*,\mathbf{u}^*)=g_i(\mathbf{x}^*)\le 0$.   We need to consider two cases: (i) $\partial_{u_i}L(\mathbf{x}^*,\mathbf{u}^*)=0$ in which case $(\mathbf{x}^*,\mathbf{u}^*)$ is a fixed point of (\ref{eq:dyn2}) and (ii) $\partial_{u_i}L(\mathbf{x}^*,\mathbf{u}^*)<0$ in which case by complementary slackness $u_i^*=0$ and this is also a fixed point of (\ref{eq:dyn2}).   Hence, every $(\mathbf{x}^*,\mathbf{u}^*)$ is a fixed point of the dynamics (\ref{eq:dyn1})-(\ref{eq:dyn2}).
\end{proof}

 \subsection{Convergence}
Let $V(\mathbf{x},\mathbf{u}):=  \min_{\mathbf{x}^*\in X^*,\mathbf{u}^*\in U(\mathbf{x}^*)}\sum_{j=1}^n(x_j-x_j^*)^2 + \sum_{i=1}^m (u_i-u_i^*)^2$.   Observe that (i) $V(\mathbf{x},\mathbf{u})\ge 0$ and (ii) $V(\mathbf{x},\mathbf{u})= 0$ if and only if $\mathbf{x}\in X^*$ and $\mathbf{u}\in U(\mathbf{x}^*)$.

\begin{lemma}\label{th:conv1}
Under update (\ref{eq:dyn1})-(\ref{eq:dyn2}), 
\begin{IEEEeqnarray*}{rCl}
&&V(\mathbf{x}(t+1),\mathbf{u}(t+1)) \le \nonumber \\
&&\bigg[V(\mathbf{x}(t),\mathbf{u}(t)) -2\alpha \left( L(\mathbf{x}(t),\mathbf{u}^*(t)) - L(\mathbf{x}^*(t),\mathbf{u}(t)) \right) \nonumber\\
&& \quad \: + \alpha^2(t) \epsilon(\mathbf{x}(t),\mathbf{u}(t))\bigg]
\end{IEEEeqnarray*}
where $\epsilon(\mathbf{x},\mathbf{u}) = \sum_{j=1}^n\left(\partial_{x_j}L(\mathbf{x},\mathbf{u})\right)^2 + \sum_{i=1}^m g_i^2(\mathbf{x})$ and $(\mathbf{x}^*(t),\mathbf{u}^*(t))=\arg \min_{\mathbf{x}^*\in X^*,\mathbf{u}^*\in U(\mathbf{x}^*)}\sum_{j=1}^n(x_j(t)-x_j^*)^2 + \sum_{i=1}^m (u_i(t)-u_i^*)^2$.
\end{lemma}
\begin{proof}
From (\ref{eq:dyn1}), for any $x^*\in X^*$ we have
\begin{align}
&\sum_{j=1}^n \left(  x_j(t+1)-x_j^*\right)^2 \nonumber\\
&= \sum_{j=1}^n \left(x_j(t)-x_j^* - \alpha\partial_{x_j}L(\mathbf{x}(t),\mathbf{u}(t)) \right)^2 \nonumber\\
 &=\sum_{j=1}^n(x_j(t)-x_j^*)^2 + 2\alpha \sum_{j=1}^n(x_j^*-x_j(t)) \partial_{x_j}L(\mathbf{x}(t),\mathbf{u}(t)) \nonumber\\
 & \quad \: +\alpha^2(t) \sum_{j=1}^n\left(\partial_{x_j}L(\mathbf{x}(t),\mathbf{u}(t))\right)^2 \nonumber\\
 &\stackrel{(a)}{\le} \sum_{j=1}^n(x_j(t)-x_j^*)^2 + 2\alpha \left(L(\mathbf{x}^*,\mathbf{u}(t)) -L(\mathbf{x}(t),\mathbf{u}(t))\right) \nonumber\\
& \quad \: +\alpha^2(t) \sum_{j=1}^n\left(\partial_{x_j}L(\mathbf{x}(t),\mathbf{u}(t))\right)^2 \label{eq:x}
\end{align}
where $(a)$ follows from the fact that $L(\mathbf{x}^*,\mathbf{u})-L(\mathbf{x},\mathbf{u})  \ge \sum_{j=1}^n\left(x_j^*-x_j  \right)\partial_{x_j} L(\mathbf{x},\mathbf{u})$ (from the definition of a subgradient).  From (\ref{eq:dyn2}) we have for any $\mathbf{u}^*\in U(\mathbf{x}^*)$ that
\begin{align}
& \sum_{i=1}^m\left(u_i(t+1)-u_i^*\right)^2 = \sum_{i=1}^m\left( u_i(t) + \alpha g_i(\mathbf{x}(t)) -u_i^* \right)^2 \nonumber\\
 &\qquad  =\sum_{i=1}^m\left(u_i(t)-u_i^*\right)^2 \nonumber \\
 & \qquad \qquad + 2\alpha  \sum_{i=1}^m \bigg( (u_i(t)-u_i^*) g_i(\mathbf{x}(t)) + \alpha^2(t) g_i^2(\mathbf{x}(t))\bigg) \nonumber\\
 & \qquad \stackrel{(a)}{=}   \sum_{i=1}^m\left(u_i(t)-u_i^*\right)^2 + 2\alpha   \left(L(\mathbf{x}(t),\mathbf{u}(t)) - L(\mathbf{x}(t),\mathbf{u}^*) \right) \nonumber \\
 & \qquad \qquad + \alpha^2(t) g_i^2(\mathbf{x}(t)) \label{eq:y}
\end{align}
where $(a)$ follows from the observation that  
\begin{IEEEeqnarray*}{rCl}
L(\mathbf{x},\mathbf{u}) - L(\mathbf{x},\mathbf{u}^*) &=& f(\mathbf{x}) + \sum_{i=1}^m u_i g_i(\mathbf{x}) - f(\mathbf{x}) - \sum_{i=1}^m u_i^* g_i(\mathbf{x}) \nonumber \\
&=& \sum_{i=1}^m(u_i-u_i^*) g_i(\mathbf{x})
\end{IEEEeqnarray*}
Then, from (\ref{eq:x}) and (\ref{eq:y}),
\begin{align*}
&V_{\mathbf{u}^*}(\mathbf{x}(t+1),\mathbf{u}(t+1))  \nonumber\\
& \quad \leq  \sum_{j=1}^n(x_j(t+1)-x_j^*(t))^2 + \sum_{i=1}^m (u_i(t+1)-u_i^*(t))^2  \\
&\quad  \leq   V_{\mathbf{u}^*}(\mathbf{x}(t),\mathbf{u}(t)) -2\alpha \left( L(\mathbf{x}(t),\mathbf{u}^*(t)) - L(\mathbf{x}^*(t),\mathbf{u}(t)) \right)\nonumber\\
&\qquad  \quad + \alpha^2(t) \epsilon(\mathbf{x}(t),\mathbf{u}(t))
\end{align*}
\end{proof}

\begin{lemma}\label{th:main}
Under update (\ref{eq:dyn1})-(\ref{eq:dyn2}), when $\frac{1}{t}\sum_{\tau=0}^t  \epsilon(\mathbf{x}(\tau),\mathbf{u}(\tau))\le M$ (e.g. this holds when $(\mathbf{x}(\tau),\mathbf{u}(\tau))$ is bounded and $f(\mathbf{x})$, $g(\mathbf{x})$ are continuous) we have
\begin{IEEEeqnarray*}{rCl}
0 & \le & \frac{1}{t}\sum_{\tau=0}^t \left( L(\mathbf{x}(\tau),\mathbf{u}^*(\tau)) - L(\mathbf{x}^*(\tau),\mathbf{u}(\tau)) \right)  \nonumber\\
& \leq & \frac{1}{2\alpha t}V(\mathbf{x}(0),\mathbf{u}(0)) + \frac{\alpha M}{2}
\end{IEEEeqnarray*}
where $\mathbf{x}^*(\tau)\in X^*$ and $\mathbf{u}^*(\tau)\in U^*(\mathbf{x}^*(\tau))$.
\end{lemma}
\begin{proof}
By Lemma \ref{th:conv1},
\begin{align*}
&V(\mathbf{x}(t+1),\mathbf{u}(t+1)) - V(\mathbf{x}(0),\mathbf{u}(0)) \\
&\qquad \le \sum_{\tau=0}^t \Big( 
-2\alpha \big( L(\mathbf{x}(\tau),\mathbf{u}^*(\tau))- L(\mathbf{x}^*(\tau),\mathbf{u}(\tau)) \big) \\
&\qquad \qquad \qquad \qquad+ \alpha^2(t) \epsilon(\mathbf{x}(\tau),\mathbf{u}(\tau))
\Big)
\end{align*}
Hence,
\begin{align*}
&\frac{1}{t}\sum_{\tau=0}^t \left( L(\mathbf{x}(\tau),\mathbf{u}^*(\tau)) - L(\mathbf{x}^*(\tau),\mathbf{u}(\tau)) \right)  \nonumber \\
&\qquad \le \frac{1}{2\alpha t}V(\mathbf{x}(0),\mathbf{u}(0)) + \frac{\alpha}{2t}\sum_{\tau=0}^t  \epsilon(\mathbf{x}(\tau),\mathbf{u}(\tau))
\end{align*}
For $x^*\in X^*$, $\mathbf{u}^*\in U(\mathbf{x}^*)$ recall $\mathbf{u}^*= \arg \sup_{\mathbf{u}\ge 0} L(\mathbf{x}^*,\mathbf{u})$ and $\mathbf{x}^* = \arg \inf_{\mathbf{x}} L(\mathbf{x},\mathbf{u}^*)$.  Hence, $L(\mathbf{x}^*,\mathbf{u}) \le L(\mathbf{x}^*,\mathbf{u}^*) \le L(\mathbf{x},\mathbf{u}^*)$ and $L(\mathbf{x},\mathbf{u}^*) - L(\mathbf{x}^*,\mathbf{u})\ge  0$.   Therefore, $L(\mathbf{x}(\tau),\mathbf{u}^*(\tau)) - L(\mathbf{x}^*,\mathbf{u}(\tau))\ge 0$.   Substituting for $\frac{1}{t}\sum_{\tau=0}^t  \epsilon(\mathbf{x}(\tau),\mathbf{u}(\tau))\le M$ then yields the result.
\end{proof}

\begin{biography}{Bahar Partov} is pursuing a PhD degree at Hamilton Institute together with Bell-labs Alcatel-Lucent Ireland. She received her master's degree from University of Essex at 2009. She did her undergraduate degree at University of Tabriz, Iran. Her current research interests are distributed algorithms for self-organized networks.
\end{biography}
\begin{biography}{Doug Leith} graduated from the University of Glasgow in 1986 and was awarded his PhD, also from the University of Glasgow, in 1989. In 2001, Prof. Leith moved to the National University of Ireland, Maynooth to assume the position of SFI Principal Investigator and to establish the Hamilton Institute (www.hamilton.ie) of which he is Director.  His current research interests  include the analysis and design of network congestion control and  resource allocation in wireless networks.
\end{biography}
\begin{biography}{Rouzbeh Razavi} 
 is a member of technical staff in the
Autonomous Networks and Systems
Research Department at Alcatel-Lucent Bell
Labs Ireland and United Kingdom. He
received his B.Sc. degree in electrical and
electronics engineering from Khajeh Nasir
Toosi University of Technology, Tehran, Iran. He also
received a master’s degree with distinction in
telecommunication and information systems and a
Ph.D. in real time multimedia communication over
wireless networks, both from the University of Essex,
United Kingdom. At Bell Labs, Dr. Razavi's current
research work involves developing algorithms for large
scale, distributed, self-organizing networks for the next
generation of wireless networks (4G and beyond) and
small cell flat cellular networks. He has published more
than 42 technical papers in peer reviewed journals and
conferences and has authored five book chapters.
\end{biography}
\end{document}